\documentclass[showkeys,aps,nofootinbib,onecolumn,floatfix]{amsart}

\usepackage{amsmath,amsthm,amsfonts,amssymb,times,bbm,graphicx}
\usepackage{verbatim}
\usepackage[norelsize]{algorithm2e}

\usepackage{color}

\usepackage{hyperref}

\newtheorem{theorem}{Theorem}
\newtheorem{proposition}[theorem]{Proposition}
\newtheorem{lemma}[theorem]{Lemma}
\newtheorem{corollary}[theorem]{Corollary}
\newtheorem{definition}[theorem]{Definition}

\def\ZZ{\mathbbm{Z}}
\def\RR{\mathbbm{R}}
\def\CC{\mathbbm{C}}
\def\R{\mathbbm{R}}

\def\EE{\mathbbm{E}}
\def\AA{\mathcal{A}}
\def\ee{\epsilon}
\def\II{\mathcal{I}}

\def\MM{\mathcal{M}}
\def\op{\mathrm{op}}

\def\Id{\mathbbm{1}}

\def\RR{\mathcal{R}}
\def\PP{\mathcal{P}}
\def\II{\mathcal{I}}

\DeclareMathOperator{\Sym}{Sym}

\DeclareMathOperator{\tr}{tr}
\DeclareMathOperator{\Tr}{Tr}

\makeatletter
\g@addto@macro{\endabstract}{\@setabstract}
\newcommand{\authorfootnotes}{\renewcommand\thefootnote{\@fnsymbol\c@footnote}}
\makeatother

\begin{document}

\begin{center}
  \LARGE 
  Improved Recovery Guarantees for Phase Retrieval from Coded Diffraction Patterns \par \bigskip

\normalsize
\authorfootnotes
D.\ Gross\textsuperscript{1,2}, 
F.\ Krahmer\textsuperscript{3},
R.\ Kueng\footnote{Corresponding author: richard.kueng@physik.uni-freiburg.de}\textsuperscript{2,4} \par \bigskip

\small
\textsuperscript{1}Institute for Theoretical Physics, University of
Cologne\par
\textsuperscript{2}Institute for Physics \& FDM, University of Freiburg\par
\textsuperscript{3}Institute for Numerical and Applied Mathematics, 
University of G{\"o}ttingen\par
\textsuperscript{4}ARC Centre for Engineered Quantum Systems, School
of Physics, The University of Sydney\par
 \bigskip
  \today
\end{center}

\begin{abstract} 
In this work we analyze the problem of phase retrieval from Fourier
measurements with random diffraction patterns.
To this end, we consider the recently introduced PhaseLift algorithm,
which expresses the problem in the language of convex optimization.
We provide recovery guarantees which
require $\mathcal{O}(\log^2 d)$ different diffraction patterns, thus
improving on recent results by Cand\`es et
al.~\cite{candes_masked_2013}, which require $\mathcal{O}(\log^4 d)$
different patterns.
\end{abstract}

\section{Introduction}

\subsection{The problem of phase retrieval}	\label{sub:general_intro}

In this work we are interested in the problem of \emph{phase retrieval} which is of considerable importance in many different areas of science, where capturing phase information is hard or even infeasible.
Problems of this kind occur, for example, in X-ray crystallography, diffraction imaging, and astronomy. 

More formally, \emph{phase retrieval} is the problem of recovering an
unknown complex vector $x \in \CC^d$ from \emph{amplitude}
measurements
\begin{equation}
y_i = | \langle a_i, x \rangle |^2 \quad i = 1, \ldots, m,	\label{eq:amplitude_measurements}
\end{equation}
for a given set of measurement vectors $a_1,\ldots, a_m \in \CC^d$. 
The observations $y$ are insensitive to a global phase change
$x \mapsto e^{i\phi} x$ 
-- hence  in the following, notions 
like ``recovery'' or ``injectivity'' are always implied to
mean ``up to a global phase''.
Clearly, the most fundamental question is: Which families of
measurement vectors $\{a_i\}$ allow for a recovery of $x$ in
principle? I.e., 
for which measurements is the map $x \mapsto y$ defined by
(\ref{eq:amplitude_measurements}) injective?

Approaches based on algebraic geometry 
(for example \cite{balan_signal_2006,heinosaari_quantum_2013}) have
established
that for determining $x$, $4d+{o}(1)$ \emph{generic} measurements are sufficient 
and $4d - \mathcal{O}(\log d)$  such observations are necessary.
Here, ``generic'' means that the measurement ensembles for which the
property fails to hold lie on a low-dimensional subvariety of the algebraic
variety of all tight measurement frames.

This notion of generic success, however, is mainly of theoretical
interest. Namely, injectivity alone neither gives an indication on
how to recover the unique solution, nor is there any chance to
directly generalize the results to the case of noisy measurements. It
should be noted, however, that recently the notion of injectivity has
been refined to capture aspects of stability with respect to noise
\cite{EM12}.
 
Paralleling these advances, there have been various attempts to find tractable recovery algorithms that yield recovery guarantees. Many of these approaches are based on a linear reformulation in matrix space, which is well-known in convex programming. 
The crucial underlying observation is that the quadratic constraints (\ref{eq:amplitude_measurements}) on $x$ are linear in the outer product $X=x x^*$:
\begin{equation*}
y_i = | \langle a_i, x \rangle |^2 = \tr \left( (a_i a_i^*) X \right). \label{eq:lift}
\end{equation*}
Balan et al.~\cite{balan_painless_2009} observed that for the right
choice of $d^2$ measurement vectors $a_i$, this linear system in the
entries of $X$ admits for a unique solution, so the problem can be
explicitly solved using linear algebra techniques. This approach,
however, does not make use of the low-rank structure of $X$, which is why the required number of measurements is so much larger than what is required for injectivity.

The \emph{PhaseLift} algorithm proposed by Cand\`es et
al.~\cite{candes_phase_2013,candes_phaselift_2012,candes_solving_2012}
uses in addition the property that $X$ is of rank one, so even when
the number of measurements is smaller than $d^2$ and there is an
entire affine space of matrices satisfying
\eqref{eq:lift}, $X$ is the solution of smallest rank.
While finding the smallest rank solution of a linear system is, in
general, NP hard, there are a number of algorithms known to recover
the smallest rank solution provided the system satisfies some
regularity conditions. The first such results were based on convex
relaxation (see, for example, 
\cite{recht_guaranteed_2010, candes_power_2010, gross_recovering_2011}). PhaseLift is also based on
this strategy. For measurement vectors drawn independently at random
from a Gaussian distribution, the number of measurements required to
guarantee recovery with high probability was shown to be of optimal
order, scaling linearly in the dimension
\cite{candes_phaselift_2012,candes_solving_2012} -- see also \cite{kueng_low_2014} for a comparable statement valid for recovering matrices of arbitrary rank.
A generalized version of this result---valid for projective measurements onto random subspaces rather than random vectors---was established in \cite{bachoc_signal_2012}.
Moreover, Ref.~\cite{pohl_phase_2013} even identifies a deterministic, explicitly engineered 
set of $4d-4$ measurement vectors and proves that PhaseLift will
successfully recover generic signals from the associated measurements.
Conversely, any complex vector is uniquely determined by $4d-4$ generic phaseless measurements \cite{conca_algebraic_2014}.

Since these first recovery guarantees for the phase retrieval problem, recovery guarantees have been proved for a number of more efficient algorithms closer to the heuristic approaches typically used in practice. For example, in \cite{alexeev_phase_2012}, an approach based on polarization is analyzed and in \cite{NJS13}, the authors study an alternating minimization algorithm. In both works, recovery guarantees  are again proved for Gaussian measurements.
Further numerical approaches have been proposed and studied in \cite{ehler_quasi_2014}.

To relate all these results to practice, the structure of applications
needs to be incorporated into the setup, which corresponds to reducing
randomness and considering structured measurements. For PhaseLift, the
first partial derandomization has been provided by the authors of this
paper, considering measurements sampled from spherical designs, that
is, polynomial-size sets 
which generalize the notion of a tight frame to higher-order
tensors
\cite{gross_partial_2013}. Recently, this result has been considerably improved in \cite{kueng_low_2014}. Arguably, these derandomized measurement setups are still mainly of theoretical interest. 

A structured measurement setup closer to applications is that of coded diffraction patterns. These correspond to the composition of diagonal matrices and the Fourier transform and model the modified application setup where diffraction masks are placed between the object and the screen as originally proposed in \cite{fannjiang2012phase}. The first recovery guarantees from masked Fourier measurements were provided for polarization based recovery \cite{bandeira_phase_2013}, where the design of the masks is very specific and intimately connected to the recovery algorithm. The required number of masks is $\mathcal{O}(\log d)$, which corresponds to $\mathcal{O}(d\log d)$ measurements.

For the PhaseLift algorithm, recovery guarantees from masked Fourier
measurements were first provided in \cite{candes_masked_2013}. The
results require  ${\mathcal{O}}(d\log^4 d)$ measurements and hold with
high probability when the masks are chosen at random, which is in line
with the observation from \cite{fannjiang2012phase} that random
diffraction patterns are particularly suitable. 

In this paper, we consider the same measurement setup as
\cite{candes_masked_2013}, but improve the bound on the required
number of measurements to ${\mathcal{O}}(d\log^2 d)$.

\section{Problem Setup and Main Results}

\subsection{Coded diffraction patterns}
\label{sec:coded}

As in \cite{candes_masked_2013}, we will work with the following
setup:

In every step, we collect the magnitudes of the discrete Fourier
transform (DFT) of a random  modulation of the unknown signal $x$.
Each such modulation pattern is modeled by a random diagonal matrix.
More formally, for $\omega:= \exp \left(\frac{2 \pi i}{d} \right)$ a
$d$-th root of unity and $\{e_1, \ldots, e_d\}$ the standard basis of
$\CC^d$, denote by 
\begin{equation}\label{eqn:1dft}
	f_k = \sum_{j=1}^d \omega^{jk} e_j
\end{equation}
the 
$k$-th discrete Fourier vector, normalized so that each entry has unit modulus.
Furthermore, consider the diagonal matrix
\begin{equation}
D_l = \sum_{i=1}^d \ee_{l,i} e_i e_i^* 	\label{eq:D}
\end{equation}
where the $\ee_{l,i}$'s are independent copies of a real-valued\footnote{
	Ref.~\cite{candes_masked_2013} also included a strongly related
	model where $\epsilon$ is a complex random variable. We have
	opted to keep $\epsilon$ real,
	which implies that the $D_l$ are hermitian. This, in turn, has
	allowed us to slightly
	simplify notation throughout. 
} random variable $\ee$ which obeys
\begin{eqnarray}
\EE [ \ee ] &=& \EE [ \ee^3 ] = 0, 	\nonumber \\ 
| \ee | & \leq & b
\quad \textrm{almost surely for some $b > 0$},	\label{eq:apriori_constraint} \\
\EE [ \ee^4 ] &=& 2 \; \EE [ \ee^2 ]^2 
\quad \textrm{and we define} \quad \nu:= \EE \left[ \ee^2 \right] .	\label{eq:moment_constraint}
\end{eqnarray}
Then
 the measurements are given by
\begin{equation}
y_{k,l} = | \langle f_k, D_l x \rangle |^2 \quad 1 \leq k \leq d, \quad 1 \leq l \leq L.	\label{eq:measurement_process}
\end{equation}
It turns out (Lemma~\ref{lem:near_isotropic} below) that 
condition \eqref{eq:moment_constraint} on $\ee$ ensures that the
measurement ensemble
forms a spherical $2$-design, which draws a connection to \cite{balan_painless_2009} and \cite{gross_partial_2013}.

As an example, the criteria above 
include the  model
\begin{equation} \label{eq:mask}
	\ee \sim
	\begin{cases}
	\sqrt{2} & \textrm{with prob. } 1/4,	\\
	0 & \textrm{with prob. } 1/2,	\\
	-\sqrt{2} & \textrm{with prob. } 1/4.
	\end{cases}
\end{equation}
which has been discussed in \cite{candes_masked_2013}.
In this case, each modulation is given by a Ra\-de\-mach\-er vector
with random erasures.

\subsection{Convex Relaxation}

Following \cite{balan_painless_2009}, we rewrite the measurement constraints as the inner product of two rank $1$ matrices, one representing the signal, the other one the measurement coefficients. In the coded diffraction setup, we obtain, as in \cite{candes_masked_2013}, that the inner product of 
 (\ref{eq:measurement_process}) can be translated into matrix form by applying the following ``lifts'':
\begin{equation*}
X :=  x x^*	
\quad \textrm{and} \quad
F_{k,l} := D_l f_k f_k^* D_l.
\end{equation*}
Occasionally, we will make use of the representation with respect to
the standard basis, which reads
\begin{equation}
F_{k,l} 
= \sum_{i,j=1}^d \ee_{l,i} \ee_{l,j} \omega^{k(i-j)} e_i e_j^*.	\label{eq:Fkl}
\end{equation}

With these definitions, the $dL$ individual linear measurements assume
the following form
\begin{eqnarray*}
y_{k,l} &=& 
\Tr \left(F_{k,l} X \right) \quad k=1,\ldots,d,\;1 \leq l \leq L.
\end{eqnarray*}
and the phase retrieval problem thus becomes the problem of finding rank 1
solutions $X=x x^*$ compatible with these affine constraints.
Rank-minimization over affine spaces is NP-hard in general. However,
it is now well-appreciated 
\cite{recht_guaranteed_2010, candes_power_2010,
gross_recovering_2011,candes_phaselift_2012}
that nuclear-norm based convex relaxations solve this problems
efficiently in many relevant instances.
Applied to phase retrieval, the relaxation becomes
\begin{eqnarray}
\textrm{argmin}_{X'} & &  \| X' \|_1	\label{eq:convex_program}\\
\textrm{subject to} & & \tr \left( F_{k,l} X' \right) = y_{k,l} \quad k = 1,\ldots n, \; 1 \leq l \leq L,	\nonumber \\
& & X' = \left(X' \right)^*	\nonumber	\\ 
& & X' \geq 0, \nonumber
\end{eqnarray}
which has been dubbed \emph{Phaselift} by its inventors
\cite{candes_phase_2013,candes_phaselift_2012,candes_solving_2012}. 
For this convex relaxation, recovery guarantees are known for
measurement vectors drawn i.i.d.\ at random from a Gaussian
distribution \cite{candes_phaselift_2012,candes_solving_2012},
$t$-designs \cite{gross_partial_2013, kueng_low_2014}, or in the masked Fourier setting \cite{candes_masked_2013}. 

We want to point out that access to additional information can considerably simplify Phaselift.
In particular, knowledge of the  signal's \emph{intensity} $y_0 = \| x \|_{\ell_2}^2$ results in an additional trace constraint which together with $X' \geq 0$ 
implies $\| X' \|_1 = y_0$ for any feasible $X'$. 
Consequently, minimizing the nuclear norm becomes redundant and \eqref{eq:convex_program} can be replaced by the feasibility problem
\begin{eqnarray}
\textrm{find} & & X' \label{eq:feasibility_program} \\
\textrm{subject to} & & \tr \left( F_{k,l} X' \right) = y_{k,l} \quad k = 1,\ldots n, \; 1 \leq l \leq L,	\nonumber \\
& & X' = \left(X' \right)^*	\nonumber	\\ 
& & \tr (X') = y_0,	\nonumber \\
& & X' \geq 0. \nonumber
\end{eqnarray}

\subsection{Our contribution}

In this paper, we adopt the setup from \cite{candes_masked_2013}.  Our
main message is that recovery of $x$ can be guaranteed already for
\begin{equation*}
  L \geq C \log^2 d
\end{equation*}
random diffraction patterns,
provided that the signal's intensity $y_0 = \| x \|_{\ell_2}^2$ is known\footnote{
This can, for instance, be achieved by starting the measurement
process with a trivial modulation pattern---i.e. $D_0$ corresponds to
the identity matrix---and summing up the $d$ corresponding
measurements \eqref{eq:measurement_process}.}.
This improves the
bound given in
\cite{candes_masked_2013}
by a factor of $\mathcal{O}(\log^2 d)$. 
It is significant, as it indicates
that the provably achievable rates are  approaching the
ultimate limit. 
Indeed, for the Rademacher masks
with random erasures introduced above, a lower bound for the number
of diffraction patterns required to allow for recovery with any algorithm is
given by $\mathcal{O}(\log d)$. This follows from a standard coupon collector's argument similar to the ones provided in
\cite{candes_power_2010,
gross_recovering_2011}. For completeness, the lower bound is precisely formulated and proved in Lemma~\ref{lem:lower} in the appendix.

Thus there cannot be a recovery algorithm requiring fewer than
$O(\log d)$ masks and there is only a single $\log$-factor separating
our results from an asymptotically tight solution.

More precisely, our version of \cite[Theorem 1.1]{candes_masked_2013}
reads:

\begin{theorem}[Main Theorem]	\label{thm:main_theorem}
	Let $x \in \CC^d$ be an unknown signal with $\|x\|_{\ell_2}=1$
	and let $d \geq 3$ be an odd
	number. 
	Suppose that $L$ complete Fourier measurements using independent
	random diffraction patterns (as defined in Section~\ref{sec:coded})
	are performed.

	Then with probability at least $(1-
	\mathrm{e}^{-\omega})$ Phaselift (the convex optimization problem
	\eqref{eq:convex_program} endowed with the additional constraint $\tr (X') = 1$, or the feasibility problem \eqref{eq:feasibility_program}) recovers $x$ up to a global phase, provided
	that 
	\begin{equation*} 
		L \geq C \omega \log^2 d.  
	\end{equation*} 
	Here, $\omega \geq 1$ is an arbitrary parameter and $C$ a
	dimension-independent
	constant that can be explicitly bounded.
\end{theorem}

The number $C$ is of the form $C=\tilde C \frac{b^8}{\nu^4}
\log_2^2 \left( b^2/\nu\right)$, where $b$ and $\nu$ were defined in \eqref{eq:apriori_constraint} and \eqref{eq:moment_constraint}, respectively. 
Also, $\tilde C$ an absolute constant for which
an explicit estimate can be extracted from our proof.

For the benefit of the technically-minded reader, we briefly sketch
the relation between the proof techniques used here, as compared to
References~\cite{candes_masked_2013} and \cite{gross_partial_2013}.
\begin{itemize}
	\item
	The general structure of this document closely mimics
	\cite{gross_partial_2013} (which bears remarkable similarity to
	\cite{candes_masked_2013}, even though the papers were written
	completely independently and with different aims in mind).

	\item
	From \cite{candes_masked_2013} we borrow the use of Hoeffding's
	inequality to bound the probability of ``the inner product
	between the measurement vectors and the signal becoming too large''.
	This is Lemma~\ref{lem:undesired_events} below. Our previous work
	also bounded the probability of such events
	\cite[Lemma~13]{gross_partial_2013}---however in a weaker way
	(relying only on certain $t$th moments as opposed to a Hoeffding
	bound).

	\item
	Both \cite{gross_partial_2013,candes_masked_2013} as well as the
	present paper estimate the condition number of the measurement operator
	restricted to the tangent space at $x x^*$ (``robust injectivity'').
	Our Proposition~\ref{prop:robust_injectivity} improves over
	\cite[Section 3.3]{candes_masked_2013} by using an operator
	Bernstein inequality instead of a weaker operator Hoeffding bound. 
	\item
	Finally, we use a slightly refined version of the golfing scheme to
	construct an approximate dual certificate (following
	\cite[Section~III.B]{gross_recovering_2011}).
\end{itemize}

\subsection{More general bases and outlook}
\label{sec:outlook}

The result allows for a fairly general distribution of the masks
$D_l$, but refers specifically to the Fourier basis.  An obvious
question is how sensitively the statements depend on the properties of
this basis.

We begin by pointing out that Theorem~\ref{thm:main_theorem}
immediately implies a corollary for higher-dimensional Fourier
transforms. In diffraction imaging applications, for example, one
would naturally employ a 2-D Fourier basis
\begin{equation}\label{eqn:2dft}
	f_{k,l} = \sum_{i=1}^{d_x} \sum_{j=1}^{d_y} \omega_{d_x}^{ik}
	\omega_{d_y}^{jl} e_{i,j},
\end{equation}
with
$d_x$ and $d_y$ the horizontal and vertical resolution respectively, 
$\omega_d := \exp\left(\frac{2\pi i}{d}\right)$, and $e_{i,j}$ the
position space basis vector representing a signal located at
coordinates $(i,j)$. 
Superficially,
(\ref{eqn:2dft}) looks quite different
from the one-dimensional case (\ref{eqn:1dft}). However, a basic
application of the Chinese Remainder Theorem shows 
that if $d_x$ and $d_y$ are co-prime, then the 2-D transform reduces to the
1-D one for dimension $d_x d_y$ (in the sense that the respective
bases agree up to relabeling) 
\cite{good_interaction_1958}. 
An analogous result holds for
higher-dimensional transforms
\cite{good_interaction_1958}, proving the following corollary.

\begin{corollary}
	Assume
	$d=\prod_{i=1}^k d_i$ is the product of mutually co-prime odd
	numbers greater than $3$. Then Theorem~\ref{thm:main_theorem}
	remains valid for the $k$-dimensional Fourier transform
	over $d_1, \dots, d_k$.
\end{corollary}

More generally speaking, our argument employs the particular
properties of Fourier bases in two places:
Lemma~\ref{lem:near_isotropic} and Lemma~\ref{lem:aux1}. 

The former lemma shows that the measurements are drawn from an  
\emph{isotropic ensemble} (or \emph{tight frame}) in the relevant
space of hermitian matrices. A similar condition is frequently used in
works on phase retrieval, low-rank matrix completion, and compressed
sensing
(e.g.\ \cite{
gross_partial_2013, 
candes_masked_2013,
rudelson_reconstruction_2013,
kueng_ripless_2013,
gross_recovering_2011}). 
Properties of the Fourier basis are used in the proof of 
Lemma~\ref{lem:near_isotropic} only for concreteness. Using relatively
straight-forward representation theory, one can give a far more
abstract version of the result which is valid for any basis satisfying
two explicit polynomial relations (cf.\ the remark below the lemma).
The combinatorial structure of Fourier transforms is immaterial at
this point.

This contrasts with Lemma~\ref{lem:aux1} which currently prevents us
from generalizing the main result to a broader class of bases. Its
proof uses explicit coordinate expressions of the Fourier basis to
facilitate a series of simplifications.  Identifying the abstract gist
of the manipulations is the main open problem which we hope to address
in future work.

We make use of the condition that $d$ be odd only for
Lemma~\ref{lem:near_isotropic}. While that particular Lemma fails to
hold for even dimensions, we find it plausible that the result as a
whole remains essentially true for even dimensions.

It would also be interesting to use the techniques of the present
paper to re-visit the problem of quantum state tomography
\cite{gross_quantum_2010, flammia_quantum_2012,
schwemmer_efficient_2014, gross_focus_2013} (which was the initial
motivation for one of the authors to become interested in low-rank
recovery methods).  Indeed, the original work on quantum state
tomography and low-rank recovery \cite{gross_quantum_2010} was based
on a model where the expectation value of a Pauli matrix is the
elementary unity of information exctractable from a quantum
experiment. While this correctly describes some experiments, it is
arguably more common that the statistics of the eigenbasis of an
observable are the objects that can be physically directly accessed. For this
practically more relevant case, no recovery guarantees seem to be
currently known and the methods used here could be used to amend that
situation.

\section{Technical Background and Notation}

\subsection{Vectors, Matrices, and matrix valued Operators}\label{sec:VMO}

The signals $x$ are assumed to live in  $\CC^d$ equipped with the
usual inner product $\langle \cdot, \cdot \rangle$. We denote the
induced norm by
\begin{equation*}
\| z \|_{\ell_2} = \sqrt{ \langle z, z \rangle } \quad \forall z \in \CC^d.
\end{equation*}
Vectors in $\CC^d$ will be denoted by lower case Latin characters. For
$z \in \CC^d$ we define the absolute value $|z|\in\R_+^d$
component-wise $|z|_i=|z_i|$.  

On the level of matrices we will exclusively encounter $d \times d$ hermitian matrices and denote them by capital Latin characters. 
Endowed with the Hilbert-Schmidt (or Frobenius) scalar product
\begin{equation}
( Z, Y ) = \tr (ZY)	\label{eq:hilbert_schmitt}
\end{equation}
the space $H^d$ of all $d \times d$ hermitian matrices becomes a
Hilbert space itself. In addition to that, we will require three different
operator norms
\begin{eqnarray}
	\| Z \|_1 &=& \tr (|Z|)\quad \textrm{(trace or nuclear norm)}, 	\nonumber \\
	\|Z \|_2 &=& \sqrt{\tr (Z^2)}  \quad \textrm{(Frobenius norm)},	\nonumber \\
	\| Z \|_\infty &=& 
	= \sup_{y \in \CC^d} \frac{| \langle y, Z y \rangle |}{\|y\|_{\ell_2}^2}	\label{eq:operator_norm}
	\quad \textrm{(operator norm)}.
\end{eqnarray}
In the definition of the trace norm, $|Z|$ denotes the unique positive semidefinite matrix obeying $|Z|^2 = Z^2$ (or equivalently $|Z| = \sqrt{ Z^2}$ which is unique).
For arbitrary matrices $Z$ of rank at most $r$, the norms above are
related via the inequalities
\begin{equation*}
\| Z \|_2 \leq \| Z \|_1 \leq \sqrt{r} \| Z \|_2
\quad \textrm{and} \quad
\| Z \|_\infty \leq \| Z \|_2 \leq \sqrt{r} \| Z \|_\infty. 
\end{equation*}
Recall that a hermitian matrix $Z$ is positive semidefinite if one has
$\langle y, Z y \rangle \geq 0$ for all $y \in \CC^d$.
We write $Y \geq Z$ iff $Y-Z$ is positive semidefinite.

In this work,  hermitian rank-1 projectors are of particular importance.
They are of the form $Z = z z^*$ with $z\in\CC^d$. 
The vector $z$ can then be recovered from $Z$ up to a global
phase factor via the singular value decomposition.
In this work, the most prominent rank-$1$ projectors are $X=xx^*$ and $F_{k,l}= D_l f_k (D_l f_k)^*$.

Finally, we will also encounter \emph{matrix-valued operators} acting
on the matrix space $H^d$. 
Here, we will restrict ourselves to operators that are hermitian with respect to the Hilbert-Schmitt inner product.
We label such objects with calligraphic letters. The operator norm
becomes
\begin{equation}
	\| \MM \|_{\op} 
	= \sup_{Z \in H^d} \frac{| \tr (Z \MM Z ) |}{\| Z \|_2^2}.
	\label{eq:self_dual_Frobenius}
\end{equation}
It turns out that only two classes of such
operators will appear in our work, namely the identity map
\begin{eqnarray*}
\II : H^d & \to & H^d 		\\
Z & \mapsto & Z \quad \forall Z \in H^d
\end{eqnarray*}
and (scalar multiples of) projectors onto some matrix $Y \in H^d$ as given by
\begin{eqnarray*}
\Pi_Y: H^d & \to & H^d 		\\
Z & \mapsto& Y (Y,Z) = Y \tr(Y Z) \quad \forall Z \in H^d.
\end{eqnarray*}
An important example of the latter class is 
\begin{equation*}
\Pi_{\Id}: \; Z \mapsto \Id \tr (\Id Z ) = \tr (Z) \Id \quad \forall Z \in H^d.
\end{equation*}
Note that 
the normalization is such that
$\frac{1}{d} \Pi_{\Id}$ 
is idempotent, i.e.\ a properly normalized projection.
Indeed, for $Z \in H^d$ arbitrary it holds that
\begin{equation}
(d^{-1} \Pi_{\Id} )^2 Z = d^{-2} \Id \tr (\Id \Pi_{\Id} Z ) = d^{-2} \tr (\Id) \tr (Z) \Id = d^{-1} \Pi_{\Id} Z.	\label{eq:P_Idaux}
\end{equation}

The notion of positive-semidefiniteness directly translates to matrix
valued operators.  It is easy to check that all the operators
introduced so far are positive semidefinite. From (\ref{eq:P_Idaux})
we obtain the ordering
\begin{equation}
	0 \leq \Pi_{\Id} \leq d \II. 	\label{eq:P_Id_bound}
\end{equation}

\subsection{Tools from Probability Theory}

In this section, we recall some concentration inequalities which will
prove useful for our argument.  Our first  tool is a slight extension  of
Hoeffding's inequality \cite{ho63}. 

\begin{theorem}	\label{thm:hoeffding}
Let $z = (z_1,\ldots, z_d) \in \CC^d$ be an arbitrary vector and let $\ee_i$, $i=1,\ldots d$, be independent copies of a real-valued, centered random variable $\ee$ which is almost surely bounded in modulus by $b > 0$.
Then
\begin{equation}
\Pr \left[ \left| \sum_{i=1}^d \ee_i z_i \right| \geq t \| z\|_{\ell_2} \right] \leq 4 \exp \left( - t^2 /(4b^2) \right) 	\label{eq:hoeffding}.
\end{equation}
\end{theorem}

One way to prove this statement, is to split up $z$ into $x + i y$ with $x,y \in \mathbb{R}^d$ and noting that $\| z \|_{\ell_2} \geq \left( \| x \|_{\ell_2} + \| y \|_{\ell_2} \right)/\sqrt{2}$ holds. Splitting up the sum into real and imaginary parts, applying the triangle inequality and bounding $\Pr \left[ \left| \sum_{i=1}^d \epsilon_i x_i \right| \geq t \| x \|_{\ell_2}/ \sqrt{2} \right]$ and $\Pr \left[ \left| \sum_{i=1}^d \epsilon_i y_i \right| \geq t \| y \|_{\ell_2} / \sqrt{2} \right]$ individually by means of Hoeffding's inequality (or a slightly generalized version of \cite[Corllary 7.21]{foucart_mathematical_2013})
then establishes \eqref{eq:hoeffding} via the union bound.

Secondly, we will require two matrix versions of Bernstein's
inequality. Such matrix valued large deviation bounds have been
established first in the field of quantum information by Ahlswede and
Winter \cite{ahlswede_strong_2002} and introduced to 
sparse and low-rank recovery in \cite{gross_quantum_2010,
gross_recovering_2011}.
We make use of refined versions from
\cite{tropp_user_2012,tropp_introduction_2012}, see also \cite[Chapter 8.5]{foucart_mathematical_2013} for the former. 
Note that as $H^d$ is a finite dimensional vector space, the results
also apply to matrix valued operators as introduced in
section~\ref{sec:VMO}.

\begin{theorem}[Uniform Operator Bernstein inequality, \cite{tropp_user_2012,gross_recovering_2011}] \label{thm:bernstein}
Consider a finite sequence $\left\{M_k \right\}$ of independent random self-adjoint  matrices.
Assume that each $M_k$ satisfies $\EE \left[ M_k \right] = 0$ and $\|M_k \|_\infty \leq \overline{R}$ (for some finite constant $\overline{R}$) almost surely.
Then with the variance parameter $ \sigma^2 := \| \sum_k \EE \left[ M_k^2 \right] \|_\infty$, the following chain of inequalities holds for all $t\geq 0$.
\begin{equation}
\Pr \left[ \Big\| \sum_k M_k \Big\|_\infty \geq t \right] \leq d \; \exp \left( - \frac{t^2/2}{\sigma^2 + \overline{R}t/3} \right) 
\leq 
\begin{cases}
 d \, \exp ( - 3t^2 / 8 \sigma^2 ) & t \leq \sigma^2/\overline{R}	\\ 
d \, \exp (-3t/8\overline{R}) & t \geq \sigma^2/\overline{R} .
\end{cases}
\label{eq:matrix_bernstein}
\end{equation}
\end{theorem}

\begin{theorem}[Smallest Eigenvalue Bernstein Inequality, \cite{tropp_introduction_2012}]	\label{thm:smallest_eigenvalue_bernstein}
Let $S = \sum_k M_k$ be a sum of i.i.d. random matrices $M_k$ which obey $\EE \left[ M_K \right] = 0$ and $\lambda_{\textrm{min}}(M_k) \geq - \underline{R}$ almost surely for some fixed $\underline{R}$.
With the variance parameter
$\sigma^2 (S) = \| \sum_k \EE \left[ M_k^2 \right] \|_\infty $
the following chain of inequalities holds for all $t \geq 0$.
\begin{equation*}
\Pr \left[ \lambda_{\min} (S) \leq - t \right] \leq d \exp \left( - \frac{t^2/2}{\sigma^2 + \underline{R}t/3} \right) 
\leq
\begin{cases}
 d \, \exp ( - 3t^2 / 8 \sigma^2 ) & t \leq \sigma^2/\underline{R}\\ 
d \, \exp (-3t/8 \underline{R}) & t \geq \sigma^2/\underline{R}. 
\end{cases}
\end{equation*}
\end{theorem}

Finally, we are also going to require a type of vector Bernstein inequality.
Note that, since $H^d$ is a $d^2$-dimensional real vector space, the statement remains valid for a sum of random hermitian matrices.

\begin{theorem}[Vector Bernstein inequality] \label{thm:vector_bernstein}
Consider a finite sequence $\{M_k \}$ of independent random vectors. 
Assume that each $M_k$ satisfies $\EE \left[ M_k \right] = 0$ and $\| M_k \|_2 \leq B$ (for some finite constant $B$) almost surely. 
Then with the variance parameter $\sigma^2 := \sum_{k} \EE \left[ \| M_k \|_2^2 \right]$,
\begin{equation*}
\Pr \left[ \left\| \sum_k M_k \right\|_2 \geq t \right] \leq \exp \left( - \frac{t^2}{4 \sigma^2} + \frac{1}{4} \right)
\end{equation*}
holds for any $t \leq \sigma^2 / B$. 
\end{theorem}

This particular vector-valued Bernstein inequality is based on the exposition in \cite[Chapter 6.3, equation (6.12)]{ledoux_probability_1991} and a direct proof can be found in \cite{gross_recovering_2011}.

\section{Proof Ingredients}

\subsection{Near-isotropicity}

In this section we study the \emph{measurement operator}\footnote{
We are going to use the notations $\MM (Z)$ and $\MM Z$ equivalently.
} 
\begin{eqnarray}
\RR : H^d & \to & H^d ,		
\quad
\RR := \sum_{l=1}^L \MM_l \quad \textrm{with} \label{eq:RR}\\
\MM_l Z 
&:=& \frac{1}{\nu^2dL} \sum_{k=1}^d \Pi_{F_{k,l}} Z = \frac{1}{\nu^2dL} \sum_{k=1}^d F_{k,l} \tr (F_{k,l} Z ),	\label{eq:MM}	
\end{eqnarray}
which just corresponds to $\RR = \frac{1}{\nu^2dL} \AA^* \AA$, where $\nu$ was defined in (\ref{eq:moment_constraint}).

The following result shows that this operator is \emph{near-isotropic}
in the sense of \cite{gross_partial_2013, candes_phase_2013}.

\begin{lemma}[$\RR$ is \emph{near-isotropic}] \label{lem:near_isotropic}
The operator $\RR$ defined in (\ref{eq:RR}) is \emph{near-isotropic} in the sense that
\begin{equation}
\EE [ \RR ] 
= L \EE \left[ \MM_l \right]
= \II + \Pi_{\Id}
\quad \textrm{or } \quad
\EE \left[ \RR (Z) \right] = Z + \tr (Z) \Id \quad \forall Z \in H^d. 		\label{eq:near_isotropic}
\end{equation}
\end{lemma}

A proof of Lemma \ref{lem:near_isotropic} can be found in
\cite{candes_masked_2013}.  However, we still present a proof -- which
is of a slightly different spirit -- in the appendix for the sake of
being self-contained. 

Two remarks are in order with regard to the previous lemma.

First, it is worthwhile to point out that \emph{near-isotropicity} of
$\RR$ is equivalent to stating that the set of all possible
realizations of $D_l f_k$ form  a 2-design.
This has been made explicit recently in
\cite[Lemma 1]{appleby_group_2013}. The notion of higher-order
spherical designs is the basic mathematical object of our previous
work \cite{gross_partial_2013} on phase retrieval.

Second, our proof of Lemma~\ref{lem:near_isotropic} uses the explicit
representation of the measurement vectors with  respect to the
standard basis. As alluded to in Section~\ref{sec:outlook}, a more
abstract proof can be given. 
We sketch the basic idea here and refer the reader to an upcoming work
for details \cite{majenz_preparation_2014}.
Consider the case where $\epsilon$ is a symmetric random variable (i.e.,
where $\epsilon$ 
has the same distribution as $-\epsilon$). 
In that case, the distribution of the $D_l$ is plainly invariant under
permutations of the main diagonal elements and under element-wise sign
changes. These are the symmetries of the $d$-cube. They constitute the
group $\ZZ_2^d \rtimes S_d$, sometimes refered to as the
\emph{hyperoctahedral group}. Using a standard technique
\cite{gross_evenly_2007, kueng_stabilizer_2013}, conditions for
near-isotropicity can be phrased in terms of the representation theory
of the hyperoctahedral group acting on $\operatorname{Sym}^2(\CC^d)$.
This action decomposes into three explicitely identifiable irreducible
components, from which one can deduce that near-isotropicity holds for
any basis that fulfillls two 4th order polynomial equations 
\cite{majenz_preparation_2014}. 

Let now $x \in \CC^d$ be the signal we aim to recover. Since the
intensity of $x$ (i.e., its $\ell_2$-norm)  is known by assumption, we can w.l.o.g. assume that $\| x \|_{\ell_2} = 1$. As in \cite{candes_phaselift_2012,gross_partial_2013,candes_masked_2013} 
we consider the space
\begin{equation}
T := \left\{ xz^* + z x^*: \; z \in \CC^d \right\} \subset H^d 	\label{eq:T}
\end{equation}
which is the tangent space of the manifold of all rank-1 hermitian matrices at the point $X = x x^*$.
The orthogonal projection onto this space can be given explicitly:
\begin{eqnarray}
\PP_T: H^d & \to & H^d 		\nonumber	\\
Z & \mapsto & XZ +ZX - XZX 	\label{eq:PP_T1} \\
&=& XZ +ZX - \tr (XZ) X. 	\label{eq:PP_T2}
\end{eqnarray}
The Frobenius inner product allows us to define an ortho-complement $T^\perp$ of $T$ in $H^d$. We denote the projection onto $T^\perp$ by $\PP_T^\perp$
and decompose any matrix $Z \in H^d$ as
\begin{equation*}
Z = \PP_T Z + \PP_T^\perp Z =: Z_T + Z_T^\perp.
\end{equation*} 
We point out that, in particular,
\begin{equation}
\PP_T \Pi_{\Id} \PP_T = \Pi_X 	
\quad \textrm{and} \quad \| \PP_T Z \|_\infty \leq 2 \| Z \|_\infty \label{eq:aux11}
\end{equation}
holds for any $Z \in H^d$. 
The first fact follows by direct calculation, while the second one comes from
\begin{equation*}
\| Z_T \|_\infty = \| Z - Z_T^\perp \|_\infty \leq \| Z \|_\infty + \| Z_T^\perp \|_\infty \leq 2 \| Z \|_\infty
\end{equation*}
where the last estimate used the pinching inequality
\cite{bhatia_matrix_1997} (Problem II.5.4).

\subsection{Well-posedness/Injectivity}

In this section, we follow
\cite{candes_phaselift_2012,gross_recovering_2011,
candes_masked_2013} in order to establish a certain injectivity
property of the measurement operator $\AA$.  

Our Proposition~\ref{prop:robust_injectivity} is the analogue of
Lemma~3.7 in \cite{candes_masked_2013}. The latter contained a
factor of $\mathcal{O}(\log^2d)$ in the exponent of the failure
probability, which does not appear here.  The reason is that we employ
a single-sided Bernstein inequality, instead of a symmetric Hoeffding
inequality. 

\begin{proposition}[Robust injectivity, lower bound]		\label{prop:robust_injectivity}
With probability of failure smaller than $d^2 \exp \left( - \frac{ \nu^4 L}{C_1 b^8} \right)$ the inequality
\begin{equation}
\frac{1}{\nu^2 dL} \| \AA (Z) \|^2_{\ell_2} > \frac{1}{4} \| Z \|^2_2 \label{eq:robust_injectivity}
\end{equation}
is valid for all matrices $Z \in T$ simultaneously. Here $b$ and $\nu$ are as in (\ref{eq:apriori_constraint}, \ref{eq:moment_constraint}) and $C_1$ is an absolute constant.
\end{proposition}

We require bounds on certain variances for the
proof of this statement. The technical Lemma \ref{lem:aux1} serves this
purpose.

\begin{lemma} \label{lem:aux1}
Let $Z \in T$ be an arbitrary matrix and let $\MM_l$ be as in (\ref{eq:MM}).
Then it holds that
\begin{equation}
\left\| \EE \left[ \MM_l(Z)^2 \right] \right\|_\infty \leq \frac{30 b^8}{\nu^4 L^2} \| Z \|_2^2, 	\label{eq:aux1}
\end{equation}
and
\begin{equation}
\left\| \EE \left[  (\PP_T \MM_l (Z))^2 \right] \right\|_\infty 
\leq \tr \left( \EE \left[ \left( \PP_T \MM_l (Z) \right)^2 \right] \right)
\leq \frac{60 b^8}{\nu^4 L^2} \| Z \|_2^2.	\label{eq:aux2}
\end{equation}
\end{lemma}

In the following proof we will use that for $a,b \in \ZZ_d = \left\{0,\ldots,d-1 \right\}$ one has
\begin{equation}
\frac{1}{d} \sum_{k=1}^d \omega^{k (a \ominus b)} = \delta_{a,b} =
\begin{cases}
1 & \textrm{if } a = b,	\\
0 & \textrm{else}.
\end{cases}		\label{eq:kronecker}
\end{equation}
The symbols $\oplus$ and $\ominus$ denote addition and subtraction modulo $d$. 

\begin{proof}[Proof of Lemma \ref{lem:aux1}]
	Let $y,z,v \in \CC^d$ be vectors of unit length. 
	Compute:
	\begin{eqnarray}
		&&
		\nu^4 L^2
		\EE\left[
			\MM_l(yy^*) \MM_l(z z^*)
		\right] v \nonumber \\
		&=&
		\frac1{d^2}
		\sum_{k,j=1}^d 
		\EE \left[ 
		\left( \sum_{i_3,i_4=1}^d \ee_{i_3} \ee_{i_4}
		\omega^{k(i_3-i_4)} \bar y_{i_3} y_{i_4} \right)
		\left( \sum_{i_5,i_6=1}^d \ee_{i_5} \ee_{i_6} \omega^{j(i_5-i_6)}
		\bar z_{i_5} z_{i_6} \right)		\right.		
		\label{eqn:7def} 
		\\
		&\times& \left.  
		\sum_{i_1,i_2,i_7,i_8=1}^d 
		\ee_{i_1} \ee_{i_2} \omega^{k(i_2-i_1)}  
		\ee_{i_7} \ee_{i_8} \omega^{j(i_8-i_7)} e_{i_2} \delta_{i_1, i_8} v_{i_7}  \right]
		\nonumber \\ 
		&=& 
		\sum_{i_1,\ldots, i_7} 
		\EE \left[ \ee_{i_1}^2 \ee_{i_2} \cdots \ee_{i_7} \right] 
		\left( \frac{1}{d} \sum_{k} \omega^{k(i_2 +i_3 - i_1 - i_4)} \right)
		\left( \frac{1}{d} \sum_{j} \omega^{j (i_5 + i_1 - i_6 - i_7)}
		\right) \nonumber	\\
		& \times & 
		\bar y_{i_3} y_{i_4}
		\bar z_{i_5} z_{i_6}
		v_{i_7}
		\,
		e_{i_2} \nonumber \\
		&=& 
		\sum_{i_1,\ldots,i_7} \EE \left[ \ee_{i_1}^2 \ee_{i_2} \cdots \ee_{i_7} \right] 
		\delta_{i_1,(i_2 \oplus i_3 \ominus i_4)} \delta_{i_1,(i_6 \oplus i_7 \ominus i_5)}
		\bar y_{i_3} y_{i_4}
		\bar z_{i_5} z_{i_6}
		v_{i_7}
		\,
		e_{i_2} \label{eqn:7ft} \\
		&=& 
		\sum_{i_2,\ldots,i_7}
		\EE \left[ \ee_{i_2 \oplus i_3 \ominus i_4}^2 \ee_{i_2} \cdots \ee_{i_7} \right] 
		\delta_{i_2,(i_4 \oplus i_6 \oplus i_7 \ominus i_3  \ominus i_5)} 
		\bar y_{i_3} y_{i_4}
		\bar z_{i_5} z_{i_6}
		v_{i_7}
		\,
		e_{i_2}, \label{eqn:7label} 
	\end{eqnarray}
	where in (\ref{eqn:7def}) we have inserted the definition of $\MM_l$,
	in (\ref{eqn:7ft}) have made use of  
	(\ref{eq:kronecker}), and in (\ref{eqn:7label}) we have 
	eliminated $i_1$.
	We now make the crucial observation that the expectation
	\begin{equation}\label{eqn:Eeps}
		\EE \left[ \ee_{i_2 \oplus i_3 \ominus i_4}^2 \ee_{i_2} \cdots \ee_{i_7} \right] 
	\end{equation}
	vanishes unless every number in $i_2, \dots, i_7$ appears at least
	twice. More formally, the expectation is zero unless the set $\{2,
	\dots, 7\}$ can be partitioned into a disjoint union of pairs
	$
		\{2, \dots, 7\} = \bigcup_{\{k,l\} \in E} \{k,l\}
	$
	such that $i_k = i_l$ for every $\{k,l\}\in E$ (in graph theory,
	$E$ would be a set of edges constituting a
	\emph{matching}).
	Indeed, assume to the contrary that there is some $j$
	such that $i_j$ is unmatched (i.e., \ $i_j \neq i_k$ for all $k\neq
	j$). We distinguish two cases: If $i_j \neq i_2\oplus i_3 \ominus i_4$,
	then $\epsilon_j$ appears only once in the product in
	(\ref{eqn:Eeps}) and the expectation vanishes because
	$\EE[\epsilon_j]=0$ by assumption. If
	$i_j= i_2\oplus i_3 \ominus i_4$, then the same conclusion holds
	because we have also assumed that $\EE[\epsilon_j^3]=0$ (this is the
	only point in the argument where we need third moments of
	$\epsilon$ to
	vanish).

	With this insight, we can proceed to put a tight bound on the
	$\ell_2$-norm of the initial expression.
	\begin{eqnarray}
		&&
		\|
			\nu^4 L^2
			\EE\left[
				\MM(yy^*) \MM(z z^*)
			\right] v 
		\|_{\ell_2}
		\nonumber \\
		&=&
		\Big\|
			\sum_{i_2,\ldots,i_7=1}^d
		\EE \left[ \ee_{i_2 \oplus i_3 \ominus i_4}^2 \ee_{i_2} \cdots \ee_{i_7} \right] 
		\delta_{i_2,(i_4 \oplus i_6 \oplus i_7 \ominus i_3  \ominus i_5)} 
		\bar y_{i_3} y_{i_4}
		\bar z_{i_5} z_{i_6}
		v_{i_7}
		\,
		e_{i_2}
		\Big\|_{\ell_2} \nonumber \\
		&\leq&
		\Big\|
			\sum_{i_2,\ldots,i_7=1}^d
		\EE \left[ \ee_{i_2 \oplus i_3 \ominus i_4}^2 \ee_{i_2} \cdots \ee_{i_7} \right] 
		\bar y_{i_3} y_{i_4}
		\bar z_{i_5} z_{i_6}
		v_{i_7}
		\,
		e_{i_2}
		\Big\|_{\ell_2} \nonumber \\ 
		&\leq&
		\sum_{\text{matchings } E}
		\Big\|
			\sum_{
				\substack{
				i_2,\ldots,i_7\\
				i_k = i_l \text{ for } \{k,l\} \in E
				}
			} 
			\left|
			\EE \left[ \ee_{i_2 \oplus i_3 \ominus i_4}^2 \ee_{i_2} \cdots \ee_{i_7} \right] 
		\bar y_{i_3} y_{i_4}
		\bar z_{i_5} z_{i_6}
		v_{i_7}
			\right|
			\,
			e_{i_2}
		\Big\|_{\ell_2} \nonumber \\ 
		&\leq&
		b^8
		\sum_{\text{matchings } E}
		\Big\|
			\sum_{
				\substack{
				i_2,\ldots,i_7\\
				i_k = i_l \text{ for } \{k,l\} \in E
				}
			} 
			\left|
		\bar y_{i_3} y_{i_4}
		\bar z_{i_5} z_{i_6}
		v_{i_7}
			\right|
			\,
			e_{i_2}
		\Big\|_{\ell_2},  \label{eqn:7cs}
	\end{eqnarray}	
	where the three inequalities follow, in that order, by realizing
	that making individual coefficients of $e_{i_2}$ larger will
	increase the norm; restricting to non-zero expectation values as per
	the discussion above; and using the assumed bound $|\ee | \leq b$.

	Now fix a matching $E$. Let $x^{(1)}$ be the vector in $\{v, \bar y, y,
	\bar z, z\}$ whose index in (\ref{eqn:7cs}) is paired with $i_2$.
	Label the remaining four vectors in that set by $x^{(2)}, \dots,
	x^{(5)}$, in such a way that $x^{(2)}$ and $x^{(3)}$ are paired and
	the same is true for
	$x^{(4)}$ and $x^{(5)}$. Then the summand corresponding to that matching becomes
    \begin{eqnarray*}
        &&
        \|
            \sum_{a,b,c=1}^d
            \left|
                x^{(1)}_{a}
                x^{(2)}_{b}
                x^{(3)}_{b}
                x^{(4)}_{c}
                x^{(5)}_{c}
            \right|
            \,
            e_{a}
        \|_{\ell_2} \\
        &=&
        \left(
            \sum_{b=1}^d |x^{(2)}_b x^{(3)}_b|
        \right)
        \left(
        \sum_{c=1}^d |x^{(4)}_c x^{(4)}_c|
        \right)
        \left\|
            \sum_{a=1}^d
                |x^{(1)}_a| e_a
        \right\|_{\ell_2} \leq 1,
    \end{eqnarray*}   
    by the Cauchy-Schwarz inequality 
    and the fact that all the $x^{(i)}$ are of length one. As there are $15$ possible matchings of $6$ indices, we
	arrive at
	\begin{equation*}
		\left\|
			\EE\left[
				\MM(yy^*) \MM(z z^*)
			\right] v 
		\right\|_{\ell_2}
		\leq
		\frac{15b^8}{\nu^4 L^2}.
	\end{equation*}

	Finally, let $Z\in T$. As $Z$ has rank at most two, we can choose normalized vectors
	$y, z \in \CC^d$ such that
	$Z=\lambda_1 y y^* + \lambda_2 z z^*$. Then
	\begin{equation*}
		\left\|	\EE[\MM(Z)^2] \right\|_\infty
		\leq
		\sum_{i,j=1}^2 | \lambda_i | \; |\lambda_j | \, \frac{15 b^8}{ \nu^4 L^2}
		=
		\|Z\|_1^2 \frac{15b^8}{\nu^4L^2}
		\leq
		\|Z\|_2^2 \frac{30b^8}{\nu^4 L^2}.
	\end{equation*}

For (\ref{eq:aux2}) we start by noting positive-semidefiniteness of $\EE \left[ \left( \PP_T \MM_l (Z) \right)^2 \right]$ implies the first inequality. 
In order to bound the trace-term, we insert (\ref{eq:PP_T1}) for $\PP_T$, expand the product, cancel terms using $X^2 = X = x x^*$ and use cyclicity of the trace to arrive at
\begin{eqnarray*}
\tr \left( \EE \left[ \left( \PP_T \MM_l (Z) \right)^2 \right] \right) 
&=& 2 \tr \left( \EE \left[ X \MM_l (Z)^2 \right] \right) 
- \tr \left( \EE \left[ (X \MM_l (Z) )(\MM_l (Z) X ) \right] \right) \\
& \leq & 2 \tr \left( X \;\EE \left[ \MM_l (Z)^2 \right] \right)
= 2 \langle x, \EE \left[ \MM_l (Z)^2 \right] x \rangle \\
& \leq & 2 \| \EE \left[ \MM_l (Z)^2 \right] \|_\infty.
\end{eqnarray*}
The upper bound in \eqref{eq:aux2} is thus implied by \eqref{eq:aux1}. 
\end{proof}

With Lemma \ref{lem:aux1} at hand, we can proceed to the lower bound on robust injectivity.

\begin{proof}[Proof of Proposition \ref{prop:robust_injectivity}]
We strongly follow the ideas presented in \cite[Proposition 9]{gross_partial_2013} and aim to show the more general statement 
\begin{equation}
\Pr \left[ (\nu^2 dL)^{-1} \| \AA (Z) \|_{\ell_2}^2 \leq (1-\delta) \| Z \|_2^2 \quad \forall Z \in T \right] \leq d^2 \exp \left( - \frac{ \nu^4 \delta^2  L}{\tilde{C}_1 b^8} \right) \label{eq:general_robust_inj}
\end{equation}
for any $\delta \in (0,1)$, where $\tilde{C}_1$ is a numerical constant.

Pick $Z \in T$ arbitrary and use \emph{near isotropicity} (\ref{eq:near_isotropic}) of $\RR$ in order to write
\begin{eqnarray}
& & (\nu^2 dL)^{-1} \| \AA (Z) \|_{\ell_2}^2 	\nonumber \\
&=& (\nu^2 dL)^{-1} \sum_{l=1}^L \sum_{k=1}^d \left( \tr (F_{k,l} Z ) \right)^2 	
= \tr \Big( Z \frac{1}{\nu^2 dL} \sum_{l=1}^L \sum_{k=1}^d F_{k,l} \tr (F_{k,l} Z ) \Big) 	\nonumber	\\
&=& \tr (Z \RR Z ) 
= \tr \left( Z ( \RR- \EE[\RR] ) Z \right) + \tr \left( Z (\II+\Pi_{\Id} ) Z \right)	\nonumber \\
&=& \tr \left( Z \PP_T ( \RR - \EE [ \RR ] ) \PP_T Z \right) + \tr (Z^2) + \tr (Z)^2	\nonumber \\
& \geq & \tr \left(Z \PP_T(\RR-\EE[\RR]) \PP_T Z \right)+ \tr (Z^2)	 	\nonumber	\\
& \geq& (1 + \lambda_{\min} \left( \PP_T (\RR - \EE [ \RR]) \PP_T \right) ) \| Z \|_2^2, 	\label{eq:robust_inj_aux1} 
\end{eqnarray}
where we have used the fact that $\MM \geq \lambda_{\min} (\MM) \II$ for any matrix valued operator $\MM$ as well as $\PP_T Z = Z$.
Therefore it suffices to to bound the smallest eigenvalue of $\PP_T (\RR - \EE [\RR] ) \PP_T $ from below.
To this end we aim to use the Operator Bernstein inequality -- Theorem \ref{thm:smallest_eigenvalue_bernstein} -- and decompose
\begin{equation*}
\PP_T (\RR - \EE [ \RR ] ) \PP_T = \sum_{l=1}^L \left( \widetilde{\MM}_l - \EE [ \widetilde{\MM}_l ] \right)
\quad \textrm{with} \quad
\widetilde{\MM}_l = \PP_T \MM_l \PP_T,
\end{equation*}
where $\MM_l$ was defined in (\ref{eq:MM}).
Note that these summands have mean zero by construction. Furthermore (\ref{eq:aux11}) implies
\begin{eqnarray*}
- \frac{1}{\nu^2 L} \II - \frac{1}{\nu^2 L} \Pi_{X} 
&\leq& - \frac{1}{\nu^2 L} \PP_T \II \PP_T - \frac{1}{\nu^2 L} \PP_T \Pi_{\Id} \PP_T 	
= - \frac{1}{ L} \PP_T \EE [ \RR ] \PP_T 	\\
&=& - \PP_T \EE[ \MM_l ] \PP_T 
\leq \widetilde{\MM}_l  - \EE [\widetilde{\MM}_l  ],
\end{eqnarray*}
where the last inequality follows from $\widetilde{\MM}_l  \geq 0$. This yields an a priori bound
\begin{equation*}
\lambda_{\min} (\widetilde{\MM}_l - \EE[ \widetilde{\MM}_l ] ) \geq -2/(\nu^2 L) =: - \underline{R}.
\end{equation*}
For the variance we use the standard identity
\begin{equation*}
0 \leq \EE \left[ (\widetilde{\MM}_l - \EE [ \widetilde{\MM}_l ] )^2 \right] = \EE \left[ \widetilde{\MM}_l ^2 \right] - \EE \left[ \widetilde{\MM}_l  \right]^2 \leq \EE \left[ \widetilde{\MM}_l ^2 \right]
\end{equation*}
and focus on the last expression. For obtaining a bound on the total variance we are going to apply (\ref{eq:self_dual_Frobenius}) to $\| \EE [ \widetilde{\MM}_l^2 ] \|_\op$. 
To this end, fix $Z \in T$ arbitrary -- this restriction is valid, due to the particular structure of $\widetilde{\MM}_l$ --  and observe
\begin{eqnarray*}
| \tr \left( Z \; \EE \left[ \widetilde{\MM}_l ^2 \right] Z \right) | 
&=& |   \EE \left[ \tr \left( \MM_l (Z) \PP_T \MM_l (Z) \right] \right)| 		
= | \tr \left( \EE \left[ (\PP_T \MM_l (Z))^2 \right] \right)| 	\\
& \leq &  2 \| \EE \left[ (\PP_T \MM_l  (Z))^2  \right] \|_\infty 
\leq \frac{120 b^8}{\nu^4 L^2} \| Z \|_2^2.
\end{eqnarray*}
The first equality follows from inserting the definition (\ref{eq:MM}) of $\MM_l$ and rewriting the expression of interest. For the second equality, we have used the fact that $\tr (A B_T) = \tr (A_T B_T)$ for any matrix pair $A,B \in H^d$ ($\PP_T$ is an orthogonal projection with respect to the Frobenius inner product)  and the last estimate is due to (\ref{eq:aux2}) in Lemma \ref{lem:aux1}. 
Since $Z \in T$ was arbitrary, we have obtained a  bound on $\| \EE [ \widetilde{\MM}_l^2 ] \|_\op$
which in turn allows us to set $\sigma^2 := \frac{120 b^8}{\nu^4 L}$ for the variance. 
Now we are ready to apply Theorem \ref{thm:smallest_eigenvalue_bernstein} which implies
\begin{equation*}
\Pr \left[ \lambda_{\min} \left( \PP_T \left( \RR - \EE[\RR] \right) \PP_T \right) \leq - \delta \right] 
\leq d^2 \exp \left( - \frac{ \nu^4 \delta^2 L}{\tilde{C}_1 b^8}\right)
\end{equation*}
 for any $0 \leq \delta \leq 1 < 60 b^8/ \nu^2 = \sigma^2 / \underline{R}$ and $\widetilde{C}_1$ is an absolute constant. 
This gives a suitable bound on the probability of the undesired event
\begin{equation*}
\left\{ \lambda_{\min} \left( \PP_T ( \RR - \EE [ \RR] ) \PP_T \right) \leq - \delta \right\}.
\end{equation*}
If this is not the case,  (\ref{eq:robust_inj_aux1}) implies
\begin{equation*}
(dL)^{-1} \| \AA (Z) \|_{\ell_2}^2 > (1-\delta) \| Z \|_2^2 
\end{equation*}
for all matrices $Z \in T$ simultaneously. This proves (\ref{eq:general_robust_inj}) and setting $\delta = 3/4$ yields Proposition \ref{prop:robust_injectivity}
(with $C_1 = \frac{16}{9} \widetilde{C}_1 $).
\end{proof}

For our proof we will also require a uniform bound on $\| \AA (Z) \|_{\ell_2}$. 

\begin{lemma}[Robust injectivity, upper bound] \label{prop:robust_injectivity2}
Let $\AA$ be as above. Then the statement
\begin{equation}
\frac{1}{ dL} \| \AA (Z) \|_{\ell_2}^2 \leq b^4 d \| Z \|_2^2
\end{equation}
holds with probability 1 for all matrices $Z \in H^d$ simultaneously. 
\end{lemma}

\begin{proof}
	Estimate
	\begin{eqnarray*}
		\frac1{dL} 
		\|\AA(Z)\|_{\ell_2}^2 
		&=& 
		\frac1{dL} 
		\sum_{k,l} \left(\tr (f_k f_k^* D_l Z D_l)\right)^2
		\leq
		\max_{1 \leq k \leq d}\|f_k f_k^*\|_2^2\frac{1}{dL} 
		\sum_{l} \| D_l Z D_l\|_2^2 	\\
		&\leq&
		d \| D_l \|_\infty^4 \| Z \|_2^2
		\leq
		d b^4 \|Z\|_2^2,
	\end{eqnarray*}
	where the first inequality holds because
	the $f_k f_k^*$'s are mutually
	orthogonal. The second inequality follows from the fact that the Frobenius norm (and more generally: any unitarily invariant norm) is symmetric \cite[Proposition IV.2.4]{bhatia_matrix_1997} -- i.e., 
$\| A B C \|_2 \leq \| A \|_\infty \|B \|_2 \| C \|_\infty$ for any $A,B,C \in H^d$ --  and the last one is due to the a-priori bound $\| D_l \|_\infty \leq b$.
\end{proof}

\section{Proof of the Main Theorem / Convex Geometry}	\label{sec:convex_geometry}

In this section, we will 
 prove that the convex program (\ref{eq:convex_program}) indeed recovers the signal $x$ with high probability. 
A common approach to prove recovery is to show the existence of an \emph{approximate dual certificate}, which in our problem setup can be formalized by the following definition.

\begin{definition}[Approximate dual certificate] \label{def:dual_certificate}
Assume that the sampling process corresponds to (\ref{eq:measurement_process}). Then we call $Y \in H^d$ an \emph{approximate dual certificate}
if $Y \in \mathrm{range} \AA^*$ 
and
\begin{equation}
\| Y_T - X \|_2 \leq \frac{\nu}{4b^2 \sqrt{d}} 
\quad \textrm{as well as} \quad
\| Y_T^\perp \|_\infty \leq \frac{1}{2}.	\label{eq:dual_certificate}
\end{equation}
\end{definition}

The following proposition, showing that the existence of such a dual certificate indeed guarantees recovery, is just a slight variation of Proposition~12 in \cite{gross_partial_2013}. For completeness, we have nevertheless included a proof in the appendix.
\begin{proposition}	\label{prop:convex_geometry}
Suppose that the measurement gives us access to $\| x\|_{\ell_2}^2$ and $y_{k,l} = | \langle f_k, D_l x \rangle |^2$ for $1 \leq k \leq n$ and $1 \leq l \leq L$. 
Then the convex optimization (\ref{eq:convex_program}) recovers the unknown $x$ (up to a global phase),
provided that (\ref{eq:robust_injectivity}) holds and an approximate dual certificate $Y$ exists.
\end{proposition}

Proposition \ref{prop:convex_geometry} proves the Main Theorem of this paper, provided that an approximate dual certificate exists. 
A first approach to construct an approximate dual certificate is to set
\begin{equation}
Y = \RR (X) - \tr (X) \Id. 	\label{eq:ansatz}
\end{equation}
Note that any such $Y$ is indeed in the range of our measurement process and, in expectation, yields an exact dual certificate, $\EE [ Y] = X$.
One can then show using an operator Bernstein or Hoeffding inequality that $Y$ is close to its expectation, but the number of measurements required is too large to make the result meaningful.
This obstacle can be overcome using the golfing scheme, a refined construction procedure originally introduced in \cite{gross_recovering_2011}. 

A main difference between our approach and the approach in \cite{candes_masked_2013} is that the authors of that paper use Hoeffding's inequality in the golfing scheme, while we employ Bernstein's inequality. The resulting bounds are sharper, but require to estimate an additional variance parameter.

An issue that remains is that such bounds heavily depend on the worst-case operator norm of the individual summands. In this framework these are proportional to  
$ | \langle f_k, D_l x \rangle |^2$, which a priori can reach $b^2 d$ (recall that $\| f_k\|_2^2 = d$). To deal with this issue, we follow the approach from \cite{gross_partial_2013,candes_masked_2013} to condition on the event that their maximal value is not too large.

\begin{lemma}	\label{lem:undesired_events}
For $Z \in T$ abitrary and a parameter $\gamma \geq 1$ we introduce the event
\begin{equation}
U_{k,l} := \left\{ | \tr ( F_{k,l} Z ) | \leq 2^{3/2} b^2 \gamma \log d \| Z \|_2 \right\}, 	\label{eq:U}
\end{equation}
 If $D_l$ is chosen according to (\ref{eq:D}) it holds that
\begin{equation*}
\max_{1 \leq k \leq d} \Pr \left[ U_{k,l}^c \right] \leq 4  d^{-\gamma}.
\end{equation*}
\end{lemma}

In the following, we refer to $ \gamma$ as the \emph{truncation rate} (cf. \cite{gross_partial_2013}). Here, we fix
\begin{equation}
\gamma = 8 +  \log_2 \left( b^2 / \nu \right) , 	\label{eq:gamma}
\end{equation}
for reasons that shall become clear in the proofs of Propositions \ref{prop:dual_aux1} and \ref{prop:dual_aux2}. Here $b$ and $\nu$ are as in (\ref{eq:apriori_constraint}) and (\ref{eq:moment_constraint}). 

\begin{proof}[Proof of Lemma \ref{lem:undesired_events}]
Fix $Z \in T$ arbitrary and apply an eigenvalue decomposition
\begin{equation*}
Z = \lambda_1 yy^* + \lambda_2 zz^*
\end{equation*}
with normalized eigenvectors $u,v \in \CC^d$. Then one has for $1 \leq k \leq d$:
\begin{eqnarray*}
\Pr \left[ U_{k,l}^c \right]
& \leq & \Pr \left[ | \tr (F_{k,l} Z ) | \geq 2 b^2 \gamma \log d \| Z \|_1 \right] 		\\
&\leq& \Pr \left[ | \lambda_1 | | \langle f_k, D_l, y \rangle |^2 + | \lambda_2| | \langle f_k, D_l, z \rangle |^2 \geq (|\lambda_1| + | \lambda_2 | ) 2 b^2 \gamma \log d \right] 	\\
& \leq& \Pr \left[ | \langle f_k, D_l y \rangle | \geq \sqrt{2 b^2 \gamma \log d} \right] + \Pr \left[ | \langle f_k, D_l z \rangle | \geq \sqrt{2 b^2 \gamma \log d} \right],
\end{eqnarray*}
where the last inequality uses a union bound. The desired statement thus follows from 
\begin{equation*}
\Pr \left[ | \langle f_k, D_l u \rangle | \geq b \sqrt{2  \gamma \log d} \|u \|_{\ell_2} \right]  \leq 	2 d^{-\gamma} \quad \forall u \in \CC^d \; \forall 1 \leq k \leq d,
\end{equation*}
which we now aim to show. Fix $1 \leq k \leq d$ and $z = (z_1,\ldots, z_d) \in \CC^d$ arbitrary and insert the definitions of $f_k$ and $D_l$ to obtain
\begin{equation*}
| \langle f_k, D_l u \rangle | 
= | \sum_{i=1}^d \ee_i \left( \omega^{ki} u_i \right) |
= | \sum_{i=1}^d \ee_i \tilde{u}_i |.
\end{equation*}
Here we have defined $\tilde{u} = \left( \omega^k u_1, \ldots, \omega^{k(d-1)} u_{d-1}, u_d \right)$. 
Note that $\| \tilde{u} \|_{\ell_2} = \| u \|_{\ell_2} =1 $ holds and applying Theorem \ref{thm:hoeffding} therefore yields
\begin{eqnarray*}
\Pr \left[ \left| \sum_{i=1}^d \ee_i \tilde{u}_i \right| \geq b \sqrt{2 \gamma \log d}\right]	
&=& \Pr \left[ \left| \sum_{i=1}^d \ee_i \tilde{u}_i \right| \geq b\sqrt{2 \gamma \log d}\| \tilde{u} \|_2 \right] 		\\
& \leq & 2 \exp \left( - \gamma \log d \right) = 2 d^{-\gamma}.
\end{eqnarray*}
\end{proof}

This result will be an important tool to bound the probability of extreme operator norms.

\begin{definition} \label{def:truncated_R}
For $Z \in T$ arbitrary and the corresponding $U_{k,l}$ introduced in (\ref{eq:U}) we define the truncated measurement operator
\begin{equation}
\RR_Z := \sum_{l=1}^L \MM^Z_l 
\quad \textrm{with}\quad
\MM^Z_l := \frac{1}{\nu^2 dL}  \sum_{k=1}^d 1_{U_{k,l}} \Pi_{F_{k,l}}, 	\label{eq:RRz}
\end{equation}
where $1_{U_{k,l}}$ denotes the indicator function associated with the event $U_{k,l}$.
\end{definition}

We now show that in expectation, this truncated operator is close to the original one.

\begin{lemma}		\label{lem:RRz}
Fix $Z \in T$ arbitrary and let $\RR_Z$ and $\MM_l^Z$ be as in (\ref{eq:RRz}). Then
\begin{eqnarray*}
\| \EE [ \RR - \RR_Z ] \|_\op 
&\leq&
 \frac{4 b^4}{\nu^2} d^{2-\gamma},	\\
\| \EE [ \left( \MM_l (W) \right)^2 -  (\MM^Z_l (W)) ^2  ] \|_\infty 	
&\leq& \frac{8 b^8}{\nu^4 L^2} d^{4- \gamma} \| W \|_\infty^2, \\
\EE \left[ \left\| \MM_l - \MM_l^Z \right\|_{\op}^2 \right]
& \leq & \frac{4 b^8}{\nu^4 L^2} d^{4-\gamma}.
\end{eqnarray*}
for any $W \in H^d$.
\end{lemma}

\begin{proof}
Note that $\EE \left[ \RR \right] = L \EE [ \MM_l]$ as well as $\EE[ \RR_Z ] = L \EE \left[ \MM_l^Z \right]$.
For the first statement, we can therefore fix $1 \leq l \leq L$ arbitrary and consider $L \| \EE [ \MM_l - \MM_l^Z ] \|_\op$. Due to Jensen's inequality this expression is majorized by
$L \EE \left[ \left\| \MM_l - \MM_l^Z \right\|_\op \right]$. 
Inserting the definitions and applying Lemma \ref{lem:undesired_events} then yields the first estimate via
\begin{eqnarray*}
L \EE \left[ \| \MM_l - \MM_l^Z \|_\op \right]	
& \leq & \frac{1}{\nu^2 d} \EE \left[ \sum_{k=1}^d (1- 1_{U_{k,l}})\left\| \Pi_{F_{k,l}} \right\|_\op \right]
\leq \frac{b^4 d^2}{\nu^2 d} \sum_{k=1}^d \EE \left[ 1_{U_{k,l}^c} \right] \\
& = & \frac{b^4 d^2}{\nu^2 d }  \sum_{k=1}^d \Pr \left[ U_{k,l}^c \right]		
 \leq  \frac{b^4 d^2}{\nu^2} \max_{1 \leq k \leq d}  \Pr [ U_{k,l}^c ] 
 \leq  \frac{4 b^4}{\nu^2} d^{2-\gamma},
\end{eqnarray*}
where the second inequality is due to $\| \Pi_{F_{k,l}} \|_\op \leq b^4 d^2$ (which follows by direct calculation).
Similarly
\begin{eqnarray*}
& & \left\| \EE \left[ \left( \MM_l (W) \right)^2 - \left( \MM_l^Z (W) \right)^2 \right] \right\|_\infty 		\\
&=& \left\| \frac{1}{(\nu^2 d L)^2} \sum_{k,j=1}^d \EE \left[ (1- 1_{U_{k,l}}1_{U_{j,l}}  ) \tr (F_{k,l} W ) \tr (F_{j,l} W ) F_{k,l} F_{j,l} \right] \right\|_\infty 	\\
& \leq & \frac{1}{\nu^4 L^2 d^2} \sum_{k,j=1}^d \EE \left[  1_{U_{k,l}^c  \cup U_{j,l}^c }  | \tr (F_{k,l} W ) \tr (F_{j,l} W ) | \| F_{k,l} \|_\infty \| F_{j,l} \|_\infty \right] \\
& \leq & \frac{b^8 d^4}{\nu^4 L^2} \| W \|_\infty^2 \max_{1 \leq k,j \leq d}  \left( \Pr [ U_{k,l}^c ]  + \Pr [ U_{j,l}^c ]  \right) 	
 \leq  \frac{8b^8}{\nu^4 L^2} d^{4-\gamma} \| W \|_\infty^2
\end{eqnarray*}
Here we have used $| \tr (F_{k,l} W ) | \leq b^2 d \| W \|_\infty$ for any $W \in H^d$ and $ \| F_{k,l} \|_\infty \leq b^2 d$
(both estimates are direct consequences of the definition of $F_{k,l}$).
Finally
\begin{eqnarray*}
\EE \left[ \left\| \MM_l - \MM_l^Z \right\|_\op^2 \right]
& \leq & \frac{1}{(\nu^2 d L)^2} \EE \left[ \left( \sum_{k=1}^d (1- 1_{U_{k,l}}) \| \Pi_{F_{k,l}} \|_\op \right)^2 \right] \\
& \leq & \frac{b^8 d^4}{\nu^4 d^2 L^2} \sum_{k,j=1}^d \EE \left[ 1_{U_{k,l}^c} 1_{U_{j,l}^c} \right] 
\leq \frac{b^8 d^4}{\nu^4  L^2} \max_{1 \leq k \leq d} \Pr \left[ U_{k,l}^c \right] \\
&\leq& \frac{4 b^8 }{\nu^4 L^2} d^{4-\gamma}
\end{eqnarray*}
follows in a similar fashion.
\end{proof}

We will now establish two technical ingredients for the golfing scheme.

\begin{proposition} \label{prop:dual_aux1}
Assume $ d \geq 3$, fix $Z \in T$ arbitrary and let $\RR_Z$ be as in (\ref{eq:RRz}). 
Then 
\begin{equation}
\Pr \left[ \| \PP_T^\perp (\RR_Z (Z) - \tr (Z) \Id) \|_\infty \geq t \| Z \|_2 \right]
\leq d \exp \left( - \frac{t \nu^4 L}{C_2 b^8 \gamma \log d } \right) \label{eq:dual_aux1}
\end{equation}
for any $t \geq 1/4$ and $\gamma$ defined in (\ref{eq:gamma}). Here $C_2$ denotes an absolute constant. 
\end{proposition}

\begin{proof}
Assume w.l.o.g. that $\| Z \|_2 = 1 $.
 By Lemma \ref{lem:near_isotropic}, 
\begin{equation*}
\PP_T^\perp \EE [ \RR (Z)] = \PP_T^\perp ( Z + \tr (Z) \Id ) = 0 + \tr (Z) \PP_T^\perp \Id,
\end{equation*}
because $Z \in T$ by assumption. We can thus rewrite the desired expression as
\begin{eqnarray}
& & \| \PP_T^\perp \left( \RR_Z (Z) - \EE [ \RR (Z) ] \right) \|_\infty 		\nonumber\\
& \leq & \| \PP_T^\perp \left( \RR_Z (Z) - \EE [ \RR_Z (Z) ] \right) \|_\infty 
+ \| \PP_T^\perp \EE \left[ \RR_Z (Z) - \RR (Z) \right] \|_2 		\nonumber\\
& \leq & \| \RR_Z (Z) - \EE [ \RR_Z (Z) ] \|_\infty + \| \EE [ \RR_Z - \RR ] \|_\op \| Z \|_2 \nonumber	\\
& \leq & \| \RR_Z (Z) - \EE [ \RR_Z (Z) ] \|_\infty + t/4. \label{eq:bound_aux1}
\end{eqnarray}
In the third line, we have used that $ \| \PP_T^\perp W  \| \leq \|W \|$ for any $W \in H^d$ and any unitarily invariant norm $\| \cdot \|$ (pinching, cf.~\cite{bhatia_matrix_1997}  (Problem II.5.4)). The last inequality follows from
\begin{equation}
 \| \EE [ \RR_Z - \RR ] \|_\op \leq \frac{4 b^4}{\nu^2} d^{2-\gamma} 
\leq \frac{b^4}{ \nu^2} 2^{4-\gamma}
\leq \frac{1}{16} 
 \leq \frac{t}{4}
, \label{eq:dual_aux1_aux1}
\end{equation}
which in turn follows from Lemma~\ref{lem:RRz} and the assumptions on $d$, $t$ and $\gamma$. 
By \eqref{eq:bound_aux1}, it remains to bound the probability of the complement of the event
\begin{equation*}
E := \left\{ \| \RR_Z (Z) - \EE [ \RR_Z (Z) ] \|_\infty \leq 3t/4 \right\}
\end{equation*} 
To this end, we use the Operator Bernstein inequality (Theorem \ref{thm:bernstein}). 
We decompose
\begin{equation*}
\RR_Z (Z) - \EE [ \RR_Z (Z) ] =  \sum_{l=1}^L \left( M_l - \EE [ M_l ] \right)
\quad \textrm{with} \quad
M_l := \MM_l^Z (Z),
\end{equation*}
where $\MM_l^Z$ was defined in (\ref{eq:RRz}).
To find an a priori bound for the individual summands, we write, using that $F_{k,l} \geq 0$ holds for all $1 \leq k \leq d$,
\begin{eqnarray*}
\| M_l - \EE \left[ M_l \right]\|_\infty & \leq & \| M_l \|_\infty
+ \| \EE \left[ \MM_l (Z) - \MM_l^Z (Z) \right] \|_\infty + \| \EE \left[ \MM_l (Z) \right] \|_\infty \\
& \leq & \| M_l \|_\infty + \frac{1}{L} \| \EE \left[ \RR_l - \RR_l^Z \right] \|_\op \| Z \|_2
+ \frac{1}{L}\| Z + \tr (Z) \|_\infty \\
&\leq& \left\| \frac{1}{\nu^2 d L} \sum_{k=1}^d 1_{U_{k,l}} | \tr (F_{k,l} Z ) | F_{k,l} \right\|_\infty
+ \frac{1}{L} \left(  \frac{b^4}{\nu^2} d^{2-\gamma}+ 1 + \sqrt{2} \right) \| Z \|_2 \\
& \leq & \frac{b^4}{\nu^2 L} \left( 2^{3/2} \gamma \log d + d^{2-\gamma} + 3 \right) \| Z \|_2
\leq \frac{608 b^8 \gamma \log d}{3 \nu^4 L} =: \overline{R}. \label{eq:apriori}
\end{eqnarray*}
Here we have employed near-isotropy of $\RR$, the first estimate in Lemma \ref{lem:RRz}
and the fact that $Z \in T$ has rank at most two. 
The last but one inequality follows from $\frac{1}{d} \sum_{k=1}^d f_k f_k^*  = \Id $, $\| D_l^2 \|_\infty \leq b^2$, and $\nu \leq b^2$. 
The last estimate is far from tight, but will slightly simplify
the resulting operator Bernstein bound.
For the variance we start with the standard estimate
\begin{equation*}
\EE \left[ (M_l - \EE [ M_l] )^2 \right] 
= \EE \left[ M_l^2 \right] - \EE [ M_l ]^2 
\leq \EE \left[ M_l^2 \right]
\end{equation*}
and bound this expression via
\begin{eqnarray*}
& & \| \EE \left[ M_l^2 \right] \|_\infty
=  \left\| \EE \left[ \left(\MM_l^Z (Z)\right)^2 \right] \right\|_\infty 		\\
&\leq& \left\| \EE \left[ \left( \MM_l^Z (Z) \right)^2 - \left( \MM_l (Z) \right)^2 \right] \right\|_\infty 
+ \left\| \EE \left[ \left( \MM_l (Z)\right)^2 \right] \right\|_\infty 	\\
& \leq & \frac{8 b^8}{\nu^4 L^2} d^{4- \gamma} \| Z \|_\infty^2   + \frac{30 b^8}{\nu^4 L^2} \| Z \|_2^2 ,
\end{eqnarray*}
where we have used Lemmas \ref{lem:RRz} and  \ref{lem:aux1}. 
Using $\| Z \|_\infty\leq \| Z \|_2 = 1$ and noting that $\nu\leq b^2$ entails $\gamma = 8 + 2 \log_2 (b^2/\nu) \geq 8$  we conclude
\begin{equation*}
	\| \sum_{l=1}^L \EE [ M_l^2 ] \|_\infty
	\leq \sum_{l=1}^L \| \EE [ M_l^2 ] \|_\infty 
	\leq \frac{8 b^8}{\nu^4 L }d^{-4} + \frac{30 b^8}{\nu^4 L}
	\leq \frac{38 b^8}{\nu^4 L} =: \sigma^2.
\end{equation*}
Our choice for $\overline{R}$ now guarantees $\sigma^2 / \overline{R} = 3/(16 \gamma \log d)\leq 3t /4 $ for any $t \geq 1/4$
(here we have used $\gamma \geq 1$ and our assumption $d \geq 3$ which entails $ \log d \geq 1$).
Consequently
\begin{equation*}
\Pr \left[ E^c \right] 
= \Pr \left[ \left\| \sum_{l=1}^L \left( M_l - \EE [ M_l] \right) \right\|_\infty > 3t/4 \right]
\leq d \exp \left( - \frac{t \nu^4 L}{C_2 b^8 \gamma \log d } \right)
\end{equation*}
with $C_2$ an absolute constant. 
This completes the proof.
\end{proof}

\begin{proposition}  \label{prop:dual_aux2}
Assume $d \geq 2$ and fix $Z \in T$ arbitrary and let $\RR_Z$ be as in \eqref{eq:RRz}
with $\gamma$ defined in \eqref{eq:gamma}.
Then
\begin{equation}
\Pr \left[ \left\| \PP_T \left( \RR_Z (Z) - Z - \tr (Z) \Id \right) \right\|_2
\geq c \| Z \|_2 \right]
\leq \exp \left( - \frac{c^2 \nu^4 L}{C_3 b^8 \gamma \log d} + \frac{1}{4} \right) \label{eq:dual_aux2}
\end{equation}
holds for any $1/(2 \log d) \leq c \leq 1$.
Here, $C_3$ is again an absolute constant. 
\end{proposition}

\begin{proof}
Similar to the previous proof, we start by assuming $\| Z \|_2 = 1$ and using \emph{near-isotropy} of $\RR$ to bound the desired expression by
\begin{eqnarray*}
& & \| \PP_T \left( \RR_Z (Z) - \EE \left[ \RR (Z) \right] \right)\|_2 \\
& \leq & \| \PP_T \left( \RR_Z (Z) - \EE \left[ \RR_Z (Z) \right] \right) \|_2
+ \| \PP_T \EE \left[ \RR (Z) - \RR_Z (Z) \right] \|_2 \\
& \leq& \| \PP_T \left( \RR_Z (Z) - \EE \left[ \RR_Z (Z) \right] \right) \|_2
+ \| \PP_T \EE \left[ \RR  - \RR_Z  \right] \|_\op \| Z \|_2 \\
& \leq & \| \PP_T \left( \RR_Z (Z) - \EE \left[ \RR_Z (Z) \right] \right) \|_2
+ c/4.
\end{eqnarray*}
Here, we have used $\| \PP_T W \|_2 \leq \| W \|_2$ for any matrix $W$ (this follows e.g. from the entry-wise definition of the Frobenius norm) and a calculation similar to  (\ref{eq:dual_aux1_aux1}):
\begin{equation*}
\| \EE \left[ \RR_Z - \RR \right] \|_\op
\leq \frac{4 b^4}{\nu^2d} d^{3-\gamma}
\leq \frac{b^4}{\nu^2 \log d} 2^{5 - \gamma} \leq \frac{1}{8 \log d} \leq \frac{c}{4},
\end{equation*}
where we have used $d \geq 2$, $\gamma \geq 8$ and the assumption $c \geq 1/(2 \log d)$.
Paralleling our idea from the previous proof, we define the event
\begin{equation*}
E' := \left\{ \| \PP_T ( \RR_Z (Z) - \EE[ \RR_Z (Z) ]) \|_\infty \leq 3c/4 \right\}
\end{equation*}
which guarantees that the desired inequality is valid.
 However, in order to bound the probability of $(E')^c$, this time we are going to employ the vector Bernstein inequality---Theorem \ref{thm:vector_bernstein}. 
Decompose
\begin{equation*}
\PP_T \left( \RR_Z (Z) - \EE \left[ \RR_Z (Z) \right] \right) = \sum_{l=1}^L \left( \tilde{M}_l - \EE \left[ \tilde{M}_l \right] \right).
\end{equation*}
Note that the $\tilde{M}_l$'s are related to $M_l$ in the previous proof via
$
\tilde{M}_l = \PP_T M_l = \PP_T \MM_l^Z (Z).
$ and in particular, $\tilde{M}_l$ has at most rank two. Consequently
\begin{eqnarray*}
\| \tilde{M}_l - \EE \left[ \tilde{M}_l \right] \|_2
&\leq& \sqrt{2}\| \PP_T M_l \|_\infty + \| \PP_T \EE \left[ \MM_l^Z(Z) - \MM_l (Z) \right] \|_2 + \|\PP_T \EE \left[ \MM_l (Z) \right] \|_2\\
& \leq & 2^{3/2} \| M_l \|_\infty + \| \EE \left[ \MM_l - \MM_l^Z \right] \|_\op \| Z \|_2 
+ \frac{1}{L}\|\PP_T \left( Z + \tr (Z) \Id \right) \|_2\\
& \leq & \frac{8 b^2 \gamma \log d}{\nu^2 L} \| Z \|_2 \| D_l^2 \|_\infty + \frac{4b^4}{\nu^2L}d^{2-\gamma} \| Z \|_2
+ \frac{\| Z \|_2 + | \tr (Z) | }{L}  \\
& \leq & \frac{15 b^4 \gamma \log d}{\nu^2 L} \| Z\|_2 
 =:B,
\end{eqnarray*}
where we have used near-isotropy of $\MM_l$, 
the estimate of $\| M_l \|_\infty$ presented in \eqref{eq:apriori},
$\| \PP_T \Id \|_2 = \| X \|_2 =1 $
and $| \tr (Z) | \leq \| Z \|_1 \leq \sqrt{2} \| Z \|_2 = \sqrt{2}$, because $Z \in T$ has rank at most two.
For the variance, we estimate
\begin{eqnarray}
\EE \left[ \left\| \tilde{M}_l - \EE \left[ \tilde{M}_l \right] \right\|_2^2 \right]
&=& \EE \left[ \left\| \PP_T \left( \MM_l^Z (Z) - \EE \left[ \MM_l^Z (Z) \right] \right)\right\|_2^2 \right] \nonumber \\
& \leq & \EE \left[ \| \PP_T \MM_l (Z) \|_2^2 \right]
+ \EE \left[ \left\| \PP_T \left( \MM_l^Z (Z) - \MM_l (Z) \right) \right\|_2^2 \right] \nonumber \\
&+&  \left\| \PP_T \EE \left[ \MM_l^Z (Z) - \MM_l (Z) \right] \right\|_2^2 
+\left\| \PP_T \EE \left[ \MM_l (Z) \right] \right\|_2^2 \nonumber \\
& \leq & \EE \left[ \tr \left( \left( \PP_T \MM_l (Z) \right)^2 \right) \right]
+ \frac{1}{L^2} \left\| \PP_T \left( Z + \tr (Z) \Id \right) \right\|_2^2 \nonumber \\
&+& 2 \EE \left[ \left\| \MM_l (Z) - \MM_l^Z (Z) \right\|_\op^2 \right] \| Z \|_2^2 \nonumber \\
& \leq & \frac{60 b^8}{\nu^4 L^2} \| Z \|_2^2 + \frac{\| Z \|_2^2 + \tr (Z)^2 }{L^2}
+ \frac{8 b^8}{\nu^4 L^2} d^{4-\gamma} \| Z \|_2^2. \label{eq:vector_bernstein_aux1}
\end{eqnarray}
Applying $b^2 \geq \nu$, $\tr (Z)^2 \leq 2 \|Z \|_2^2 = 2$ and $d^{4-\gamma} \leq 1$ (because we choose $\gamma \geq 8$) 
allows us to upper-bound \eqref{eq:vector_bernstein_aux1}
by $71 b^8 / (\nu^4 L^2)$ and set
\begin{equation*}
\sum_{l=1}^L \EE \left[ \left\| \tilde{M}_l - \EE \left[ \tilde{M}_l \right] \right\|_2^2 \right]
\leq \frac{71 b^8}{\nu^4 L} \leq \frac{15  b^8 \gamma \log d}{\nu^4 L} =: \sigma.
\end{equation*}
Again, the last estimate is far from tight, but assures $\sigma^2 / B = b^4/ \nu^2 \geq 1$.
Applying the vector Bernstein inequality---Theorem \ref{thm:vector_bernstein}---for $t=3c/4$ yields the desired bound on the probability of $(E')^c$ occurring.

\end{proof}

We are now ready to construct a suitable approximate dual certificate in the sense of Definition \ref{def:dual_certificate}. 
The key idea here is an iterative procedure -- dubbed the \emph{golfing scheme} -- that was first established in \cite{gross_recovering_2011}
(see also \cite{candes_probabilistic_2011, kueng_ripless_2013, candes_masked_2013, gross_partial_2013}). 

\begin{proposition}		\label{prop:Y_construction}
Assume $d \geq 3$ and let  $\omega \geq 1$ be arbitrary. 
If the total number of $L$ of diffraction patterns
fulfills
\begin{equation}
L \geq C \omega \log^2 d,		\label{eq:dual_masks}
\end{equation}
then with probability larger than $1 - 5/6 \mathrm{e}^{-\omega}$, an approximate dual certificate 
$Y$ as in Definition \ref{def:dual_certificate} can be constructed using the golfing scheme. 
Here, $C$ is a constant that only depends on the probability distribution used to generate the random masks $D_l$. 
\end{proposition}

To be concrete, the constant $C$ depends on the truncation rate $\gamma$ -- which we have fixed in (\ref{eq:gamma}) -- and  the a-priori bound $b$ and $\nu$ of the random variable $\ee$ used to generate the diffraction patterns $D_l$:
\begin{equation}
C = \tilde{C} \gamma \frac{b^8}{\nu^4} \log_2 \left( b^2 / \nu \right)
= \bar{C} \frac{b^8}{\nu^4} \log_2^2 \left( b^2 / \nu \right)  , \label{eq:C}
\end{equation} 
where $\tilde{C}$ and $\bar{C}$ are absolute constants.

\begin{proof}[Proof of Proposition \ref{prop:Y_construction}]

\RestyleAlgo{boxruled}
\begin{algorithm}\label{alg:dual}
\caption{Pseudo-code\protect\footnotemark\  that summarizes the randomized ``golfing scheme''  for constructing an approximate dual certificate in the sense of Definition \ref{def:dual_certificate}.}
\textbf{Input}:\newline
\begin{tabular}{lll}
	\qquad & $X \in H^d$ \quad &  $\#$ signal to be recovered \\
	& $w \in \mathbb{N}$ & $\#$ maximum number of iterations (after the first two steps)\\
	& $\left\{ L_i \right\}_{i=1}^{w+2} \subset   \mathbb{N}$\quad& $\#$ number of masks used in $i$th iteration \\
	& $r$  & $\#$ require $r$ ``successful'' iterations after the first two \\
	& &    $\#$ (i.e. iterations where we enter the inner \textbf{if}-block)

\end{tabular}

\ \\
 
\textbf{Initialize:}\newline
\begin{tabular}{lll}
	\qquad & $\mathbf{Y} = [\,]$ \quad & $\#$ a list of 
	matrices in $H^d$, initially empty \\
	& $\mathbf{Q} = [X]$ \quad & $\#$ a list of matrices in $T$,
	initialized to hold $X$ as its only element \\
	& i = 1 & $\#$ number of current iteration \\
	&  $\xi = [0,\dots,0]$ & $\#$ array of $w+1$ zeros; $\xi_i$ will be set to 1 if $i$th iteration succeeds

\end{tabular}

\ \\

\textbf{Body:}\newline
\For{ $1 \leq i \leq 2$}{
set $Q$ to be the last element of $\mathbf{Q}$ and $Y$ to be the last element of $\mathbf{Y}$ \newline
Sample $L_i$ masks independently according to \eqref{eq:D} and construct $\RR_Q$
according to Def. \ref{def:truncated_R}
\eIf{\eqref{eq:dual_aux1},\eqref{eq:dual_aux2} hold for $\RR_Q$ and $Q \in T$ with parameters $t = 1/8$, $c = 1/\sqrt{2 \log d}$}
{$\xi_i = 1 $ \newline
	$Y \leftarrow  \RR_Q Q - \tr (Q) \Id+ Y$, \quad append $Y$ to $\mathbf{Y}$
	\newline
	$Q \leftarrow X - \PP_T Y$, \quad append $Q$ to $\mathbf{Q}$
	$
	i \leftarrow i+1
	$}
{abort and report \emph{failure}}
}

\While{ $3 \leq i \leq w+2$ \bf{and} $\sum_{j=3}^{i} \xi_j \leq r $}{
set $Q$ to be the last element of $\mathbf{Q}$ and $Y$ to be the last element of $\mathbf{Y}$, \newline
sample $L_{i+2}$ masks independently according to (\ref{eq:D}); construct $\RR_{Q}$
according to 
Def.~\ref{def:truncated_R}.

\If{(\ref{eq:dual_aux1}), (\ref{eq:dual_aux2}) hold for $\RR_{Q}$ and $Q \in
T$ with parameters $t=\log d /4$, $c=1/2$}{
	$\xi_i = 1 $ \newline
	$Y \leftarrow  \RR_Q Q - \tr (Q) \Id+ Y$, \quad append $Y$ to $\mathbf{Y}$
	\newline
	$Q \leftarrow X - \PP_T Y$, \quad append $Q$ to $\mathbf{Q}$
	}
	$
	i \leftarrow i+1
	$
	}
\eIf{$\sum_{i=3}^{w+2} \xi_i = r$  }
	{report \emph{success} and output $\mathbf{Y},\mathbf{Q}$, $\xi$}
	{report \emph{failure}}
\end{algorithm}
\footnotetext{
	Similar to \cite{gross_partial_2013} we use use of pseudo-code for a compact presentation of this
	randomized procedure. However, the reader should keep in mind that
	the construction is purely part of a proof and should not be
	confused with the recovery algorithm (which is given in
	Eq.~(\ref{eq:convex_program})).
}

This construction is inspired by \cite{kueng_ripless_2013,
candes_probabilistic_2011} and \cite{adcock_generalized_2011}. 
As in
\cite{gross_recovering_2011}, our construction of $Y$ follows a
recursive procedure of $w$ iterations which can be summarized in the
pseudo-code described in Algorithm~\ref{alg:dual}.
It depends on a number of parameters -- $w, L_i, r$, c.f.\ Input
section of the algorithm -- the values of which
will be chosen below.
If this algorithm succeeds, it outputs three lists
\begin{equation*}
\mathbf{Y} = \left[ Y_1,\ldots, Y_{r+2} \right], \quad \mathbf{Q} = \left[ Q_0,\ldots,Q_{r+2} \right],
\quad \textrm{and} \quad \xi = [\xi_1,\dots,\xi_{w+2} ] .
\end{equation*}
They obey iterative relations of the following form (c.f. \cite[Lemma 14]{kueng_ripless_2013}):
\begin{eqnarray*}
Y &:=&
Y_{r+2} = \RR_{Q_{r+1}} ( Q_{r+1} ) - \tr (Q_{r+1} ) \Id + Y_{r+1}  	\\
&=& \cdots = \sum_{i=1}^{r+2}  \left( \RR_{Q_{i-1}} ( Q_{i-1}) - \tr \left( Q_{i-1}\right) \Id \right) 
\quad \textrm{and} 		\\
Q_i
&=& X - \PP_T Y_i
= \PP_T \left( Q_{i-1} + \tr (Q_{i-1} ) \Id - \RR_{Q_{i-1}} (Q_{i-1}) \right) 		\\
&=& \ldots = \prod_{j=1}^i \PP_T \left( \II + \Pi_{\Id} - \RR_{Q_{j-1}} \right) Q_0.
\end{eqnarray*}
We now set
\begin{equation*}
r =  \lceil \frac{1}{2} \log_2 d \rceil + \lceil \log_2 (b^2 / \nu ) \rceil + 1
\end{equation*}
This choice, together with the validity of properties (\ref{eq:dual_aux1}) and (\ref{eq:dual_aux2}) for $t = 1/8$, $c = 1/ \sqrt{2 \log d}$ in the first two steps and for $t = \log d / 4$, $c = 1/2$ 
in each remaining update ($Y_i \to Y_{i+1}$ and $Q_i \to Q_{i+1}$, respectively) together with $Q_0 = X$ then guarantee
\begin{eqnarray*}
\| Y_T - X \|_2
&=& \| Q_{r+2} \|_2
\leq \| Q \|_0 \frac{1}{2 \log d} \prod_{i=3}^{r+2} \frac{1}{2}
= \frac{1}{\log d} 2^{-(r+1)} \leq \frac{\nu}{4b^2 \sqrt{d}}, \\
\| Y_T^\perp \|_\infty
& \leq & \sum_{i=1}^{r+2} \left\| \PP_T \left( \RR_{Q_{i-1}} (Q_{i-1}) - \tr (Q_{i-1}) \Id \right) \right\|_\infty \\
& \leq & \frac{1}{8} \| Q_0 \|_2 + \frac{1}{8} \| Q_1 \| + \sum_{i=3}^{r+2} \frac{ \log d}{4} \| Q_{i-1} \|_2 \\
& \leq & \left( \frac{1}{8} + \frac{1}{8 \sqrt{2 \log d}} + \sum_{i=3}^{r+2} \frac{\log d}{4} 
\left( \frac{1}{\sqrt{2 \log d}} \right)^2 \prod_{j=1}^{i-2} \frac{1}{2} \right) \| Q_0 \|_2 \\
& \leq & \frac{1}{4} \sum_{i=0}^\infty 2^{-i} = \frac{1}{2}
\end{eqnarray*}
which are precisely the requirements (\ref{eq:dual_certificate}) on $Y$. 

What remains to be done now is to choose parameters $w$ and $\left\{
L_i \right\}_{i=1}^{w+2}$ such that the probability of the algorithm
failing is smaller than $\frac{5}{6}\mathrm{e}^{-\omega}$.
Recall that the $\xi_i$'s are Bernoulli random variables that indicate
whether the $i$-th iteration of the algorithm failed ($\xi_i = 0$) or
has been successful ($\xi_i = 1$).
The complete Algorithm \ref{alg:dual} fails exactly if one of the
first two iterations fails
\begin{equation}\label{eqn:1stfail}
	\xi_1 = 0\qquad \text{or} \qquad \xi_2 = 0
\end{equation}
or fewer than $r$ of the remaining ones succeed
\begin{equation}\label{eqn:2ndfail}
	\sum_{i=3}^{w+2} \xi_i < r.
\end{equation}


We start by estimating the probability of (\ref{eqn:1stfail})
occuring.
Setting
\begin{equation*}
	L_1 =  L_2 = C_5 \frac{b^8}{\nu^4} \omega \gamma \log^2 d
\end{equation*}
for a sufficiently large absolute constant $C_5$, 
and using 
the union bound over Propositions \ref{prop:dual_aux1}
and \ref{prop:dual_aux2} (for $Z = X $), one obtains
\begin{eqnarray}
& & \Pr \left[ 
	\xi_1 = 0
\right] \nonumber \\
&\leq&
\Pr \left[ \textrm{(\ref{eq:dual_aux1}) fails to hold in the first step} \right]
+ \Pr \left[ \textrm{ (\ref{eq:dual_aux2}) fails to hold in the first step} \right] \nonumber\\
& \leq &  \exp \left( - \frac{(1/\sqrt{2 \log d})^2 \nu^4 L_1}{C_3 b^8 \gamma \log d} + \frac{1}{4} \right)
+ d \exp \left( - \frac{4^{-1} \nu^4 L_1}{C_2 b^8 \gamma \log d} \right) 	
 \leq  \frac{1}{6} \mathrm{e}^{-\omega}.	\label{eq:union1}	
\end{eqnarray}
An analogous bound holds for the probability of $\xi_2=0$.

We turn to (\ref{eqn:2ndfail}).
Our aim is to bound $\Pr \left[ \sum_{i=3}^{w+2} \xi_i < r \right]$ by
a similar expression involving \emph{independent} Bernoulli variables
$\xi_i'$.  To achieve this, we observe
\begin{equation*}
\Pr \left[ \sum_{i=3}^{w+2} \xi_i < r \right] = \EE \left[ \Pr \left[ \xi_{w+2}< r - \sum_{i=3}^{w+1} \xi_i \vert \xi_{w+1},\ldots,\xi_3 \right] \right]. 
\end{equation*}
Conditioned on an arbitrary instance of $\xi_{w+1},\dots,\xi_3$, the variable $\xi_{w+2}$ follows a Bernoulli distribution with some parameter $p \left( \xi_w,\ldots,\xi_2 \right)$. 
Now note that if $\xi \sim \mathrm{B}(p)$ is a Bernoulli variable with parameter $p$, then for every fixed $t \in \mathbb{R}$, the probability
$\Pr_{\xi \sim \mathrm{B}(p)} \left[ \xi < t \right]$ is non-increasing as a function of $p$. 
This observation implies that the estimate
\begin{equation}
\Pr \left[ \sum_{i=3}^{w+2} \xi_i < r \right] \leq \Pr \left[ \xi_{w+2}' + \sum_{i=3}^{w+1} \xi_i < r\right] \label{eq:david2}
\end{equation}
is valid, provided that $\xi_{w+1}'$ is an independent $p'$-Bernoulli distributed random variable with
\begin{equation*}
p' \leq \min_{\xi_{w+1},\dots,\xi_3} p \left( \xi_{w+1},\dots,\xi_3 \right).
\end{equation*}
A combination of Propositions \ref{prop:dual_aux1} and \ref{prop:dual_aux2} provides a uniform lower bound on $p \left( \xi_{w+1},\dots,\xi_3 \right)$.
Indeed, setting $Z = Q_w$ and invoking them with
\begin{equation*}
L := C_4 \frac{b^8}{\nu^4} \gamma \log d
\end{equation*}
-- where $C_4$ is a sufficiently large constant -- assures a probability of success of at least $9/10$ for any $Q$. This estimate is in particular independent of $\xi_{w+1},\dots,\xi_3$.
Consequently, by choosing $p' = 9/10$ and $L_i = L$ for all $3 \leq i \leq w+2$, we can iterate the estimate (\ref{eq:david2}) and arrive at
\begin{equation}
\Pr \left[ \sum_{i=3}^{w+2} \xi_i < r \right]
\leq \Pr \left[ \xi_{w+2}' + \sum_{i=3}^{w+1} \xi_i < r \right]
\leq \cdots
\leq \Pr \left[ \sum_{i=3}^{w+2} \xi_i' < r \right],
\end{equation}
where the $\xi_i'$'s on the right hand side are independent Bernoulli variables with parameter $9/10$. A standard one-sided Chernoff bound
(e.g. e.g \cite[Section Concentration: Theorem 2.1]{habib_probabilistic_1998}) gives
\begin{equation*}
\Pr \left[ \sum_{i=3}^{w+2} \xi_i' \leq w (9/10 - t ) \right] \leq \mathrm{e}^{-2wt^2}.
\end{equation*}
Choosing $t = 9/10 - r/w$, we then obtain
\begin{eqnarray}
\Pr \left[ \sum_{i=3}^{w+2} \xi_i' < r \right]
& \leq & \Pr \left[ \sum_{i=3}^{w+2} \xi_i' \leq r \right]
= \Pr \left[ \sum_{i=3}^{w+2} \xi_i' \leq w \left(9/10 - t \right) \right] \nonumber \\
& \leq & \exp \left( - 2w \left( \frac{9}{10} - \frac{r}{w} \right)^2 \right).
\end{eqnarray}
Setting the number of iterations generously to 
\begin{equation*}
w = 10 \omega r = 10 \omega \left( \lceil \frac{1}{2} \log_2 d \rceil + \lceil \log_2 (b^2 / \nu ) \rceil + 1 \right)
\end{equation*}
guarantees
\begin{equation*}
2 w \left( \frac{9}{10} - \frac{r}{w} \right)^2 
\geq 20 \omega r \left( 8/10 \right)^2 \geq 12 \omega r \geq \omega + \log 2,
\end{equation*}
where we have used $\omega \geq 1$ in the first and last step. 
From this estimate we can conclude
\begin{equation}
\Pr \left[ \sum_{i=3}^{w+2} \xi_i < r \right] \leq \mathrm{e}^{-\omega - \log 2} = \frac{1}{2} \mathrm{e}^{-\omega} \label{eq:union2}
\end{equation}
which suffices for our purpose. 

The desired bound of $\frac{5}{6} \mathrm{e}^{-\omega}$ on the
probability of the algorithm failing now follows from taking the union
bound over (\ref{eq:union1}) and two times (\ref{eq:union2}). 

Finally we note that with our construction the total amount of masks obeys
\begin{eqnarray*}
L 
&=& \sum_{i=1}^{w+2} L_i
= 2 C_5 \frac{b^8}{\nu^4} \omega \gamma \log^2 d + 10 \omega \left( \lceil 0.5 \log_2 d \rceil + \lceil \log_2 (b^2 / \nu \rceil \right) C_4 \frac{b^8}{\nu^4} \gamma \log d \\
&\leq& \tilde{C} \gamma\frac{b^8}{\nu^4} \log_2 \left( b^2 / \nu \right)  \omega  \log^2 d
= C \omega \log^2 d
\end{eqnarray*}
for a sufficiently large absolute constant $\tilde{C}$ (recall that we have chosen $\gamma = 8 + \log_2 \left( b^2 / \nu \right)$ in (\ref{eq:gamma}))
and $C$ as in (\ref{eq:C}).
\end{proof}

We now have all the ingredients for the proof of our main result, Theorem \ref{thm:main_theorem}.

\begin{proof}[Proof of the Main Theorem]
With probability at least $1-5/6 \mathrm{e}^{-\omega}$, the construction of Proposition~\ref{prop:Y_construction} yields an approximate dual certificate provided that
the total number of masks $L$ obeys
\begin{equation*}
L \geq \bar{C} \frac{b^8}{\nu^4} \log_2^2 \left( n^2 / \nu \right) \omega \log^2 d,
\end{equation*}
where $\bar{C}$ is a sufficiently large constant. In addition, by Proposition \ref{prop:robust_injectivity}, one has (\ref{eq:robust_injectivity}) with probability at least $1-1/6 \mathrm{e}^{-\omega}$, potentially with an increased value of $\bar{C}$. Thus the result follows from Proposition \ref{prop:convex_geometry} and a union bound over the two probabilities of failure.
\end{proof}

%
%
%

\textbf{Acknowledgements:}

DG and RK are grateful to the organizers and participants of the
Workshop on Phaseless Reconstruction, held as part of the 2013
February Fourier Talks at the University of Maryland, where they were
introduced to the details of the problem. This extends, in particular,
to Thomas Strohmer.  RK is pleased to acknolwedge extremely helpful
advice he received from Johan Aberg  troughout the course of the
project. FK thanks the organizers and participants of the AIM
workshop ``Frame theory intersects geometry'', in particular Thomas
Strohmer, G\"otz Pfander,  and Nate Strawn, for stimulating
conversations on the topic of this paper.  The authors also
acknowledge inspiring discussions with Emmanuel Can\-d\`es, Yonina
Eldar, and David James. 
We would also like to thank the anonymous referees for extremely helpful comments and suggestions which allowed us to further improve the presentation of our results.

The work of DG and RK is supported by the Excellence Initiative of the
German Federal and State Governments (Grants ZUK 43 \& 81), by
scholarship funds from the State Graduate Funding Program of
Baden-W\"urttemberg, by the US Army Research Office under contracts
W911NF-14-1-0098 and W911NF-14-1-0133 (Quantum Characterization,
Verification, and Validation), and the DFG (GRO 4334 \& SPP 1798). FK
acknowledges support from the German Federal Ministry of Education and
Reseach (BMBF) through the cooperative research project ZeMat.

\bibliographystyle{IEEEtran}
\bibliography{phaseless_bib}

\section{Appendix}
\begin{lemma}
 \label{lem:lower}
 Consider as signal the first standard basis vector $e_1\in\CC^d$. Let $a_\ell$, $\ell=1, \dots, m=Ld$. Then for every $\delta>0$ there exists $c>0$ such that the following holds for the measurement vectors corresponding to $L<c \log_2 d$ masked Fourier measurements of $e_1$ as introduced in Section~\ref{sec:coded} with random masks $\epsilon_\ell$ drawn independently at random according to the distribution given in \eqref{eq:mask}. With probability at least $1-\delta$, there exists another signal that produces the exact same measurements. Thus no algorithm will be able to distinguish these signals based on their measurements.
\end{lemma}
\begin{proof}
 As $e_1$ as well as any other standard basis vector $e_\ell$ is $1$-sparse, their phaseless measurements corresponding to one mask will just consist of the entry-wise absolute values first (or $\ell$-th, respectively) column of the corresponding masked Fourier transform matrix. As all entries of the Fourier transform matrix are of unit modulus, the measurements of $e_\ell$ are hence completely determined by the vector $v_\ell$ consisting of the $\ell$-th entry of every mask. As a consequence, $e_1$ and $e_\ell$ produce the same measurements if the entries of $v_1$ and  $v_\ell$ have the same absolute value.
 There are $L$ masks, and each entry's absolute value can be either $0$ or $\sqrt{2}$. So there are $2^L$ possible choices for $|v_\ell|$. For each $\ell>1$, one of them is drawn uniformly at random. Hence by the coupon collector's problem, a $v_\ell$ with the same absolute values as $v_1$ appears again with high probabilty within the first $\Theta(L2^L)$ draws, where by increasing the constant, one can make the probability arbitrarily small.
 For $L<c \log_2(d)$, we obtain $L 2^L < c d^c \log_2(d)$, which for $c$ small enough is less than $d-1$. Thus there will exist another $v_\ell$ with $|v_\ell|=|v_1|$, which proves the lemma.
\end{proof}

\begin{proof}[Proof of Lemma \ref{lem:near_isotropic}]
We prove formula (\ref{eq:near_isotropic}) in a way that is slightly different from the proof provided in \cite{candes_masked_2013}. We show that the set of all possible $ D_l f_k$'s is in fact proportional to a 2-design and deduce \emph{near-isotropicity} of $\RR$ from this. We refer to \cite{gross_partial_2013} for further clarification of the concepts used here.
Concretely, for $1 \leq l \leq L$ we aim to show
\begin{equation}
\frac{1}{\nu^2  d}\sum_{k=1}^d \EE \left[ F_{k,l}^{\otimes 2} \right] = 2 P_{\Sym^2}, \label{eq:near_isotropic_aux1}
\end{equation}
where $P_{\Sym^2}$ denotes the projector onto the totally symmetric subspace of $\CC^d \otimes \CC^d$.
Near isotropicity of $\RR$ directly follows from
(\ref{eq:near_isotropic_aux1}) by applying \cite[Lemma
1]{appleby_group_2013} (with $\alpha = \beta =1$):
\begin{equation*}
\EE \left[ \RR \right] Z = \frac{1}{\nu^2  dL} \sum_{k=1}^d \sum_{l=1}^L \EE \left[ F_{k,l} \tr (F_{k,l} Z ) \right]
= \frac{1}{\nu^2  d} \sum_{k=1}^d \EE \left[ F_{k,1} \tr (F_{k,1} Z) \right] = (\II + \Pi_{\Id} ) Z .
\end{equation*}
So let us proceed to deriving equation (\ref{eq:near_isotropic_aux1}). 
We do this by exploring the action of the equation's left hand side  on a tensor product $ e_i \otimes e_j$ ($1 \leq i,j \leq d$) of two standard basis vectors in $\CC^d$.
Here it is important to distinguish two special cases, namely $i=j$ and $i \neq j$.
For the former we get by inserting standard basis representations
\begin{eqnarray*}
\frac{1}{\nu^2  d} \sum_{k=1}^d \EE \left[ F_k^{\otimes 2} \right] (e_i \otimes e_i)
&=& \frac{1}{\nu^2  d} \sum_{k=1}^d \EE \left[ \ee_i^2 \langle f_k, e_i \rangle ^2 D^{\otimes 2} (f_k \otimes f_k) \right]		\\
&=& \frac{1}{\nu^2 } \sum_{a,b=1}^d \EE \left[ \ee_i^2 \ee_a \ee_b \right] \left( \frac{1}{d} \sum_{k=1}^d \omega^{k(a+b-2i)} \right) (e_a \otimes e_b) 	\\
&=& \frac{1}{\nu^2 } \sum_{a,b=1}^d \delta_{(a \oplus b),(2i)} \EE \left[ \ee_i^2 \ee_a \ee_b \right] (e_a \otimes e_b), 
\end{eqnarray*}
where we have used (\ref{eq:kronecker}) and the fact that for odd
$d$, there is a multiplicative inverse of $2$ modulo $d$.
Now
$\EE[\ee_a]=\EE[\ee_b]=0$ 
implies that one obtains a non-vanishing summand only if $a=b$.
Therefore one in fact gets
\begin{equation*}
\frac{1}{\nu^2  d} \sum_{k=1}^d \EE \left[ F_k^{\otimes 2} \right] (e_i \otimes e_i)
= \frac{1}{\nu^2 } \sum_{a=1}^d \delta_{(2a),(2i)} \EE \left[ \ee_i^2 \ee_a^2 \right] (e_a \otimes e_b)
= \frac{1}{\nu^2 } \EE \left[ \ee_i^4 \right] (e_i \otimes e_i)
= 2 (e_i \otimes e_i),
\end{equation*}
where we have used the moment condition (\ref{eq:moment_constraint}) in the last step.
This however is equivalent to the action of $2 P_{\Sym^2}$ on symmetric basis states.

Let us now focus on the second case, namely $i \neq j$. A similar calculation then yields
\begin{eqnarray*}
\frac{1}{\nu^2  d} \sum_{k=1}^d \EE \left[ F_k^{\otimes 2} \right] (e_i \otimes e_j)
&=& \frac{1}{\nu^2 } \sum_{a,b=1}^d \EE \left[ \ee_i \ee_j \ee_a \ee_b \right] \delta_{(a+b),(i+j)} (e_a \otimes e_b).
\end{eqnarray*}
Again, $\EE [ \ee ] = 0$ demands that the $\ee$'s have to ``pair up''. Since $i \neq j$ by assumption, there are only two such possibilities, namely $(i=a,j=b)$ and $(i=b,j=a)$.
Both pairings obey the additional delta-constraint and we therefore get
\begin{equation*}
\frac{1}{\nu^2  d} \sum_{k=1}^d \EE \left[ F_k^{\otimes 2} \right] (e_i \otimes e_j)
= \frac{1}{\nu^2 } \EE \left[ \ee_i^2 \ee_j^2 \right] \left( e_i \otimes e_j + e_j \otimes e_i \right)
=  (e_i \otimes e_j) + (e_j \otimes e_i),
\end{equation*}
where we have once more used (\ref{eq:moment_constraint}) in the final step. This, however is again just the action of $2 P_{\Sym^2}$ on vectors $e_i \otimes e_j$ with $i \neq j$.
Since the extended standard basis $\left\{(e_i \otimes e_j) \right\}_{1 \leq i,j \leq d}$ forms a complete basis of $\CC^d \otimes \CC^d$, we can deduce equation (\ref{eq:near_isotropic_aux1}) from this.

\end{proof}

\begin{proof}[Proof of Proposition \ref{prop:convex_geometry}]
Let $X'$ be an arbitrary feasible point of (\ref{eq:convex_program}) and we decompose it as $X' = X + \Delta$, where $\Delta$ is a feasible displacement.
Feasibility then implies $\AA (X') = \AA (X)$ and consequently $\AA (\Delta) = 0$ must hold. The pinching inequality \cite{bhatia_matrix_1997}  (Problem II.5.4) 
now implies
\begin{equation*}
\| X' \|_1 = \| X + \Delta \|_1 \geq \| X \|_1 + \tr (\Delta_T) + \| \Delta_T^\perp \|_1
\end{equation*}
and $X$ is guaranteed to be the minimum of (\ref{eq:convex_program}) if
\begin{equation}
\tr (\Delta_T ) + \| \Delta_T^\perp \|_1 > 0 	\label{eq:convex_aux1}
\end{equation}
is true for any feasible displacement $\Delta$. Therefore it suffices to show that (\ref{eq:convex_aux1}) is guaranteed to hold under the assumptions of the proposition.
In order to do so, we combine feasibility of $\Delta$ with Proposition \ref{prop:robust_injectivity} and Lemma \ref{prop:robust_injectivity2} to obtain
\begin{equation}
\| \Delta_T \|_2 < \frac{2}{\sqrt{\nu^2 dL}} \| \AA (\Delta_T) \|_{\ell_2} 
= \frac{2}{\nu \sqrt{dL}} \| \AA (\Delta_T^\perp ) \|_{\ell_2} 
\leq \frac{2b^2 \sqrt{d}}{\nu} \| \Delta_T^\perp \|_2.	\label{eq:convex_aux2}
\end{equation}
Feasibility of $\Delta$ also implies $(Y, \Delta)=0$, because $Y \in \mathrm{range}(\AA^*)$ by definition. 
Combining this insight with (\ref{eq:convex_aux2}) and the defining property (\ref{eq:dual_certificate}) of $Y$ now yields
\begin{eqnarray*}
0
&=& ( Y, \Delta) = (Y_T - X, \Delta_T) + (X, \Delta_T ) + (Y^\perp_T, \Delta_T^\perp ) 		\\
&\leq & \| Y_T - X \|_2 \| \Delta_T \|_2 + \tr (\Delta_T) + \| Y_T^\perp \|_\infty \| \Delta_T^\perp \|_1 	\\
& < & \tr (\Delta_T) + \| Y_T - X \|_2 2 b^2 \sqrt{d} /\nu \| \Delta_T^\perp \|_2 + \| Y_T^\perp \|_\infty \| \Delta_T^\perp \|_1 	\\
& \leq & \tr (\Delta_T) + 1/2 \| \Delta_T^\perp \|_2 + 1/2 \| \Delta_T^\perp \|_1 		\\
& \leq & \tr (\Delta_T) + \| \Delta_T^\perp \|_1,
\end{eqnarray*}
which is just the optimality criterion (\ref{eq:convex_aux1}). 
\end{proof}

\end{document}